\documentclass[11pt]{article}
\usepackage{fullpage,enumitem,amssymb,xcolor}
\setlist{itemsep=.2em,topsep=.5em,parsep=.2em} 
\usepackage{fullpage}

\newcommand{\para}[1]{\medskip\noindent\textbf{#1}}
\newcommand{\parafirst}[1]{\noindent\textbf{#1}}

\usepackage[english]{babel}
\usepackage[hidelinks]{hyperref}
\usepackage{physics}
\usepackage[utf8]{inputenc}
\usepackage{cleveref}

\usepackage{amsmath,amsthm,amsfonts}
\newtheorem{definition}{Definition}[section]
\newtheorem{theorem}{Theorem}[section]
\newtheorem{corollary}[theorem]{Corollary}
\newtheorem{lemma}[theorem]{Lemma}
\newtheorem{fact}[theorem]{Fact}
\usepackage{graphicx}
\usepackage{xcolor}

\newcommand{\pCst}{12}
\newcommand{\pMax}{24}
\newcommand{\E}{\mathbb{E}}
\renewcommand{\Pr}{\mathbb{P}}
\newcommand{\set}[1]{\{#1\}}
\newcommand{\est}{\textsf{{Est}}}

\newcommand{\quantumLE}{$\mathsf{QuantumLE}$}
\newcommand{\quantumAgreement}{$\mathsf{QuantumAgreement}$}
\newcommand{\quantumRWLE}{$\mathsf{QuantumRWLE}$}
\newcommand{\quantumQWLE}{$\mathsf{QuantumQWLE}$}
\newcommand{\quantumGeneralLE}{$\mathsf{QuantumGeneralLE}$}

\newcommand{\eps}{\varepsilon}

\newcommand{\setup}{\mathsf{Setup}}
\newcommand{\checking}{\mathsf{Checking}}
\newcommand{\update}{\mathsf{Update}}
\newcommand{\Sc}{\mathsf{S}}
\newcommand{\Cc}{\mathsf{C}}
\newcommand{\Uc}{\mathsf{U}}
\newcommand{\send}{\mathsf{Send}}

\newcommand{\vac}{\bot}

\usepackage{algorithm,algorithmicx,algpseudocode}
\algrenewcommand\algorithmicindent{0.5em}

\def\elected{\mbox{\small ELECTED}}
\def\nonelected{\mbox{\small NON-ELECTED}}

\title{Quantum Communication Advantage \\ for Leader Election and Agreement
}
\author{Fabien Dufoulon\thanks{School of Computing and Communications, Lancaster University, Lancaster, UK.  Email: \href{mailto: f.dufoulon@lancaster.ac.uk}{\tt f.dufoulon@lancaster.ac.uk}.}
\and Frédéric Magniez\thanks{Université Paris Cité, CNRS, IRIF, Paris, France. Email:\href{mailto:frederic.magniez@irif.fr}{\tt frederic.magniez@irif.fr}. Research supported in part by the 
European QuantERA project QOPT (ERA-NET Cofund 2022-25) and
the French PEPR integrated project EPiQ (ANR-22-PETQ-0007).}
\and Gopal Pandurangan\thanks{Department of Computer Science, University of Houston, Houston, TX 77204, USA. Email: \href{mailto: gopal@cs.uh.edu}{\tt gopal@cs.uh.edu}. Supported in part by ARO Grant W911NF-231-0191 and NSF grant CCF-2402837.}}
\date{}

\begin{document}
\maketitle

\begin{abstract}
This work focuses on understanding the quantum \emph{message} complexity of two central problems in distributed computing, namely, leader election and agreement in synchronous message-passing communication networks. We show that \emph{quantum} communication gives an advantage for both problems by presenting  quantum distributed algorithms that significantly outperform their respective classical counterparts under various network topologies. 

While prior works have studied and analyzed  quantum distributed algorithms in the context of (improving) round complexity,  a key conceptual contribution of our work is positing a framework to design and analyze the \emph{message} complexity of  quantum distributed algorithms. We present and show how quantum algorithmic techniques such as Grover search, quantum counting, and quantum walks can make distributed algorithms significantly message-efficient.  

In particular, our leader election protocol for diameter-2 networks uses \emph{quantum walks} to achieve the improved message complexity. To the best of our knowledge, this is the first such application of quantum walks in distributed computing.
\end{abstract}

\section{Introduction}
\label{sec:intro}
\emph{Message complexity} is one of the  fundamental performance measures in distributed computing, defined as the total number of messages (typically, of small size, say $O(\log n)$ bits) exchanged by all nodes throughout the operation of the distributed algorithm. Hence, it determines the communication cost of the distributed algorithm, which crucially influences other performance measures such as \emph{round complexity} 
(another fundamental measure defined as the total number of rounds of the distributed algorithm), latency, energy consumption, etc.
These measures are essential for various applications, including distributed big data computing, ad hoc sensor networks, and blockchains \cite{dongarra,message-byzantine,saia,AMP18}.  Thus, keeping the message complexity as low as possible is essential.  Indeed, there has been extensive work  
to design distributed algorithms for various problems that \emph{minimize} message complexity even at the cost of increased running time or decreased quality of solution (e.g., approximate solution) (see e.g., \cite{tcssync,congested-podc2015,podc2021,itcs2024}). However, for many fundamental distributed computing problems such as leader election, there are \emph{tight} lower bounds on the message complexity. In this paper, we show that one can significantly breach the \emph{classical} message lower bounds by taking advantage of the power of \emph{quantum} communication.

\subsection{Context}
\parafirst{Leader election and agreement.}
Leader Election is a fundamental problem in distributed computing studied well for over five decades.
In leader election, a group of processors in a distributed communication network have to elect a unique leader among themselves, i.e., only one processor must output the decision that it is the leader, say, by changing a special \emph{status} component of its state to the value \emph{leader}~\cite{Lynch_1996_Book}, and the rest must change their status to \emph{non-leader}. The non-leader nodes need not know the identity of the leader. This {\em implicit} variant of leader election is quite standard (cf.~\cite{Lynch_1996_Book}), and has been extensively studied and has numerous applications (see e.g., \cite{Lynch_1996_Book,Tel_2000_Book,Kutten_2015_JACM,KPPRT15} and the references therein).\footnote{In the  explicit variant,   all the non-leaders must also know the identity of the unique leader. Clearly, this implies a lower bound of $\Omega(n)$ messages (even quantumly), where $n$ is the network size.} 

The complexity of leader election --- both its message and round complexities --- has been extensively studied both in general graphs as well as in special graph classes such as rings, complete networks, and diameter-2 networks, see e.g. \cite{Lynch_1996_Book, Peleg_1990, Santoro_2006_Book, Tel_2000_Book, KPPRT15, Kutten_2015_JACM,ChatterjeePR20}. Note that these works (as ours) assume the standard CONGEST model (cf. Section \ref{sec:model}) where message sizes are small, typically $O(\log n)$ bits.

In~\cite{Kutten_2015_JACM}, a \emph{tight}  bound of $\Theta(m)$ messages has been shown ($m$ is the number of edges in the network) on the message complexity of leader election that applied even to Monte Carlo randomized algorithms with (large-enough) constant success probability; this lower bound applies for graphs {\em that have diameter at least three}. 
On the other hand, for complete graphs (i.e., graphs of diameter one), there exists a tight bound\footnote{$\tilde{O}$ notation hides logarithmic (in $n$) factors.}   of $\tilde{\Theta}(\sqrt{n})$
on the message complexity of randomized leader election ($n$ is the number of nodes in the network)\cite{KPPRT15}. For diameter-2 networks, a tight bound
of $\Theta(n)$ on the message complexity of leader election has been proven~\cite{ChatterjeePR20}.
  The above three results~\cite{Kutten_2015_JACM,KPPRT15,ChatterjeePR20} fully characterize the message complexity of leader election vis-\`a-vis the graph diameter in the \emph{classical} distributed setting. 

  Another fundamental problem we consider in this paper is implicit agreement, where the goal is for a \emph{non-empty subset} of nodes to agree on a common input value that should be the input value of some node~\cite{AMP18}. Implicit agreement is a generalization of leader election and the fundamental agreement problem.\footnote{In agreement, 
  \emph{all} nodes should output a common value which should be the input value of some node.}  Note that it can be solved by electing a leader who can be the only node that outputs its value.

  In~\cite{AMP18}, a tight bound (both upper and lower) of  $\tilde{\Theta}(n^{1/2})$ 
  messages for agreement was shown. This bound, as other bounds stated above for leader election, assumes that nodes have access to (only) {\em private} random bits. 
 On the other hand, the above work showed that if nodes have access to an unbiased {\em global (shared)} random bits, then   implicit agreement can be solved using $\tilde{O}(n^{2/5})$ messages (in expectation).\footnote{The same benefit does not apply to leader election~\cite{AMP18}: even with access to shared randomness, $\tilde{\Omega}(n^{0.5})$ messages  (in expectation) are needed for any {\em leader election} algorithm that succeeds with constant probability.}

 All the above results apply to the classical distributed setting. We show that all the above message bounds can be significantly improved in the \emph{quantum}  distributed setting.

\para{Quantum distributed computing.}
The study of \emph{round complexity} --- designing distributed algorithms with low round complexity (i.e., fast algorithms) and showing round lower bounds ---  has been a major focus of distributed computing both classically and  quantumly.
One of the earliest studies of round complexity in quantum distributed computing in the synchronous message-passing model 
was  in the LOCAL model, where there is no bound on the message size per round \cite{DP08,GKM09}.   Separations between the computational powers (with respect to round complexity) of the classical and quantum versions of the model have been reported for some non ``natural'' problems~\cite{GKM09,GallNR19,Balliu25}, but other papers have also reported limited improvement for other problems (e.g., approximate graph coloring~\cite{Coiteux-RoyDGKG24}).

In the quantum setting, it was shown in~\cite{ElkinKNP14}  that the quantum CONGEST model 
is not more powerful than the classical CONGEST model for many important graph-theoretical problems. Nonetheless, it was later shown in~\cite{LM18} that computing a network's diameter can be solved faster in the quantum setting. Since then, other quantum speed-ups have been discovered, in particular, for subgraph detection~\cite{CFGLO22,ApeldoornV22,FraigniaudLMT24}.

When it comes to \emph{message complexity}, as in the case of multiparty quantum communication complexity, typically \emph{oblivious} communication is assumed. In oblivious multiparty communication, the communication pattern  is pre-determined at the start of the algorithm (i.e., for every pair of nodes what rounds they will respectively communicate in). 
As mentioned in~\cite{GallS22}, oblivious communication is assumed by all quantum multiparty protocols in the literature, up to our knowledge. 
In this model, quantum advantages are known but only in the context of communication complexity, where the task is to compute a function whose inputs are distributed among several players.
Then, the complexity usually scales at least linearly with the number of players and, in the quantum context, sometimes sublinearly with the local input size. On the contrary, in distributed computing, we aim for a complexity sublinear in the number of nodes or edges, and the local input size is often negligible. Moreover, the decision is usually local. 

Finally, we note that, in the quantum distributed setting, it seems also that only oblivious communications have been studied before. Moreover, all prior works (in both the LOCAL and CONGEST models) assume all nodes, in every round, communicate messages to all their neighbors, as they are concerned with round complexity only. However, as a result, these prior works have high message complexity (measured by the total number of quantum or classical bits exchanged).

 \subsection{Our Contributions}
 
  This work focuses on the message complexity of leader election and agreement in the \emph{quantum} distributed setting. Prior works have studied and analyzed  quantum distributed algorithms in the context of (improving) round complexity, this is the first known work that addresses improving the message complexity. 
  We posit a framework to design and analyze the \emph{message complexity} of  quantum distributed algorithms and present techniques to design such algorithms that are communication-efficient. 
  Using our framework and techniques, we show that \emph{quantum} communication gives a significant advantage by presenting quantum distributed algorithms that outperform their respective classical counterparts. 
  For all the classical results for leader election and agreement mentioned above, we design quantum distributed algorithms that significantly beat the respective classical message bounds. In particular, for leader election, the quantum message bounds give significant improvements for diameter-1 (complete graphs),  diameter-2, and general graphs (diameter-3 and beyond). 

\para{New model for quantum messaging and routing.} (cf. Section \ref{sec:qdc})
As previously discussed, quantum network communication is usually oblivious when one only cares about the round complexity. The situation is quite different when it comes to the message complexity. Indeed, for problems such as leader election and agreement, the existing randomized protocols heavily rely on a \emph{non-oblivious} choice of a pattern of communication, where each node decides its local communication pattern based on its random choices and previously received messages.

We adapt and extend this non-oblivious communication behavior to the quantum setting.
We restrict ourselves to the CONGEST model, where at each round, an edge between two nodes can carry at most $O(\log n)$ bits, where $n$ is the number of  nodes in the network.
We start by allowing the possibility of controlling the message's recipient quantumly by a quantum register that could be itself in superposition. By doing so, we import to the distributed setting the notion of superpositions of trajectories~\cite{ck19}, initially defined and studied for quantum Shannon theory.

More precisely, in the context of quantum communication, quantum Shannon theory has been considered in the case of superpositions of quantum channels~\cite{ck19}, where a particle is sent in a coherent superposition of two or more transmission links, which has been further modeled by the notion of routed quantum circuits~\cite{vkb21}. Similar notions of accesses in superposition have been considered for quantum random access memory (QRAM) and quantum random access gates (QRAG). This has been considered for quantum algorithms~\cite{Amb07}
and quantum programming languages~\cite{ACCRV23}.

\para{New framework.} (cf. Section \ref{sec:subroutines})
While prior works have studied and analyzed quantum distributed algorithms in the context of (improving) round complexity,  a key conceptual contribution of our paper is positing a framework to design and analyze the \emph{message} complexity of quantum distributed algorithms.
We present and show how quantum algorithmic techniques such as Grover search, quantum counting, and quantum walks can make distributed algorithms significantly message-efficient.

One novelty of this framework is the possibility of incorporating decentralized procedures. Indeed, in the context of round complexity, all quantum routines are coordinated by a leader. By adding this possibility in the checking procedure inside Grover search, we show that further speedup is possible.
Already, for diameter-2 networks, we could get a protocol with message complexity $\tilde{O}(n^{3/4})$, using two nested Grover searches, one being centralized and the other not.
Moreover, by adding a layer of \emph{quantum walks} to Grover search, we achieve the improved message complexity of $\tilde{O}(n^{2/3})$. To the best of our knowledge, this is the first such application of quantum walks in distributed computing.

\para{New results.}
Most of our results are for \emph{leader election}, over several network configurations, and \emph{without} prior shared randomness or quantum entanglement.

\begin{enumerate}
\item For \emph{complete networks} (cf. Section \ref{sec:leader-complete}), we present a quantum leader election protocol that, with high probability\footnote{Throughout, ``with high probability" means with probability at least $1-1/n^c$ for some constant $c$.}, elects a leader and has message complexity 
    $\tilde{O}(n^{1/3})$, beating the tight $\tilde{\Omega}(\sqrt{n})$ classical bound~\cite{KPPRT15,AMP18}. 
    
\item For \emph{diameter-2 networks} (cf. Section \ref{sec:implicitLED}), we present a quantum protocol with message complexity $\tilde{O}(n^{2/3})$, beating the tight $\tilde{\Omega}(n)$ classical bound~\cite{ChatterjeePR20}.

\item For \emph{arbitrary networks} (cf. Section \ref{sec:LEonExpander}), we first show that quantum leader election in networks with mixing time $\tau$ can be solved with message complexity $\tilde{O}(\tau^{5/3} n^{1/3})$. This result assumes that nodes have knowledge of $\tau$.
In particular, if the graph has small mixing time (or high conductance) such as an expander or an hypercube where $\tau = \tilde{O}(1)$, then the above bound implies a message complexity of $\tilde{O}(n^{1/3})$.

Second (cf. Section \ref{sec:leader-general}), we show that, for any graph with $m$ edges and $n$ nodes, leader election can be accomplished with message complexity 
    $\tilde{O}(\sqrt{mn})$, which beats the tight classical bound  of $\Omega(m)$~\cite{Kutten_2015_JACM}.
   
\item  Finally, we present a result for \emph{implicit agreement} in complete networks (cf. Section \ref{sec:agreement}), when nodes are allowed to share random bits. Indeed, in this setting and classically, agreement is known to admit more efficient solutions than leader election, which itself implies agreement.
We show that a similar phenomenon occurs quantumly, and present a quantum agreement protocol for  complete networks with expected message complexity $\tilde{O}(n^{1/5})$, which improves over the best-known classical bound of $\tilde{O}(n^{2/5})$~\cite{AMP18}. 
\end{enumerate}

\para{Key ideas.}
We would like to emphasize that our framework is not sufficient on its own to speed-up existing classical results. We had to redesign them to apply our framework and to get a quantum boosting. We review our main technical and conceptual ideas through simplified but representative sub-problems.

\paragraph{Leader election and handshake}
In the randomized setting, several protocols for leader elections are based on a solution to the simpler \emph{handshake} problem: Two nodes $u,v$ would like to find a common node $w$ to speak through. 

For the case of complete networks, the randomized approach of~\cite{KPPRT15} uses the birthday paradox: $u$ and $v$ select $\Theta(\sqrt{n})$ nodes and send them a message. Clearly they will identify a common vertex $w$ with high probably and a total message complexity of $\Theta(\sqrt{n})$.
This reminds us of the collision finding problem for which there exists an efficient (sequential) quantum algorithm~\cite{BHT98} beating algorithms based on the birthday paradox. In this problem, one has to find a duplicate in a random sequence of $n$ integers in $\{1,2,\ldots,n\}$. 

We extract the main idea of this quantum algorithm, and implement it for the leader election problem, or for the simple case of the handshake, as follows. First, we break the symmetry. The protocol has now two phases, one classical and one quantum. In the first and classical phase, $u$ and $v$ contacts $k$ nodes, even deterministically, using $k$ messages. In the second and quantum phase, $u$ and $v$ use our distributed version of Grover search to find a node contacted in the first phase using $O(\sqrt{n/k})$ messages. Letting $k=n^{1/3}$ leads to the message complexity of $O(n^{1/3})$. Then, additional logarithmic factors come in the complexity to ensure high probability of success, and to adapt this idea to the original leader election problem.

Once this idea has been captured, it is not too hard to adapt it to the case of arbitrary network with mixing time $\tau$ as in~\cite{KPPRT15}. Nonetheless, there is a technical subtlety: due to the centralization of one part of Grover search, we cannot just walk on the network graph. Instead, we have to decide in advance the sequence of random choices that the walk will make. This blows up the message complexity by a factor of $\tau$, due to the propagation of $\tau$ random decisions taken by the initiator of the walk.

Probably our most challenging algorithm is for graphs of diameter $2$. Again, in that case, a basic scenario consist of \emph{handshake}, but now there might be a single node connecting $u$ to $v$. So the classical algorithm has total message complexity $\Theta(n)$~\cite{ChatterjeePR20}. This reminds us of the element distinctness problem, which is basically the worst case of collision finding where there is a single duplicate. Still quantumly, one can (sequentially) find it faster than classically using a notion of quantum walk~\cite{Amb07}. This time, the walk is not on the network, but it is used locally by the node in charge to implement and speedup Grover search.
Nonetheless, the situation is much more complicated than both our quantum protocol for complete networks and the sequential quantum algorithm for element distinctness.
First, our protocol uses several nested Grover searches, inside the main one which itself take benefit of a (local) quantum walk. Second, one of the inner Grover searches is decentralized. Both the use of a decentralized procedure and of a quantum walk are new in the context of quantum distributed computing.

\paragraph{Leader election and tree merging}
Last, we give a final leader election algorithm, for general graphs. The algorithm is based on a (standard) technique of merging trees, but it is well-known this approach (and in fact, any leader election algorithm in general networks) must send $\Omega(m)$ messages in the classical setting, as shown by~\cite{Kutten_2015_JACM}. To obtain $o(m)$ quantum message complexity, we crucially leverage Grover search to decide which trees should be merged, and this is done far more message-efficiently (quadratically so, in fact) than can be achieved in the classical setting. More concretely, we use Grover Search to find edges that connect different (adjacent) trees.

Note that a similar but non-distributed approach was taken by \cite{DHHM06}, where they give a quantum algorithm for (sequentially) finding a (minimum) spanning tree in a graph efficiently. 

\paragraph{Implicit agreement}
Our final quantum protocol is for implicit agreement. Again, we use a technique that is new, as far as we know, in the context of quantum distributed computing, but standard in sequential quantum computing. This is a variant of Grover search for approximate counting~\cite{BHT98c}. Here also, the classical approach from~\cite{AMP18} needs to be redesigned. In particular, both (1) estimating how many nodes have a certain input (or vote), and (2) detecting when agreement is reached, are redesigned to leverage quantum subroutines and obtain quadratic factor improvements in the message complexity. 

Consider the case where a node $v$ wants to estimate the number of nodes with input (or vote) 1 within some $n$-sized universe up to some $cn$ additive error, for any $c < 1$ and with at least constant success probability. Then, $v$ must contact $\Omega(1/c^2)$ nodes (and send as many messages) to achieve this estimation in the classical setting\footnote{Indeed, distinguishing a uniform random bit from a $c$-biased random bit requires $\Theta(1/c^2)$ samples by information theory arguments. Then a reduction to counting can establish the claimed bound.}. On the other hand, approximate quantum counting achieves this estimation for $v$ using only $O(1/c)$ messages, but node $v$ has no access to the nodes that lead to the estimation.

As for detecting when agreement is reached, Grover search allows us (as with the handshake problem) to get a quadratic factor reduction in the message complexity when compared to the classical setting.

\section{Preliminaries}

\subsection{Distributed Computing Model} 
\label{sec:model}
We first formally describe a standard distributed computing model in the classical setting, namely the synchronous CONGEST message-passing model (e.g., see 
\cite{peleg}). The quantum version of the model will be described in \Cref{sec:qdc}.

We consider a network of $n$ nodes, 
represented as an undirected connected graph $G=(V,E)$. We consider three types of network topologies: (1) complete graphs,  
(2) diameter-2 networks, and (3) arbitrary networks (diameter 3 and beyond).

Each node 
runs an instance of the same distributed algorithm.
The computation advances in synchronous rounds where, in every round, nodes
can send messages, receive messages that were sent in the same round by
neighbors in $G$,
and perform some local computation.
In  the CONGEST model~\cite{peleg}, a node can send in each round at most one message of size $O(\log n)$ bits per edge.

All processors have access to a {\em private} unbiased random bits. For our algorithm of  \Cref{sec:agreement}, we also
allow processors access to a {\em global (shared)} random bits.
Also, we do not assume unique identities (using private randomness, nodes can generate
unique identifiers with high probability.) Finally, throughout the paper we assume all nodes know $n$, but our results hold also when nodes only know a polynomial upper bound on $n$. 

Messages are the only means of communication; in particular, nodes
cannot access the coin flips of other nodes, and do not share any memory.
Throughout this paper, we assume that all nodes are awake initially and 
simultaneously start executing the algorithm. 

 We note that initially nodes have knowledge only of themselves, in other words we assume the {\em clean network model} --- also called the {\em KT0 model} \cite{peleg} which is standard and  commonly used. Each node $v$ has $\deg(v)$ ports that it can use to communicate respectively with its $\deg(v)$ neighbors; each  port $p=(v,u)$ of $v$ is connected
 (exclusively) to a port $p'=(u,v)$ of its neighbor $u$. Finally, we denote by $N(v)$ the neighbors of $v$ in $G$.

\subsection{Problem Definitions}\label{sec:pbs}

We first study the fundamental leader election problem.

\para{Leader Election.}
Every node $u$ has a special variable $\texttt{status}_u$ that it can set
to a value in $\{\bot, \nonelected, \elected \}$; initially we assume
$\texttt{status}_u = \bot$.
An \emph{algorithm $A$ solves leader election in $T$ rounds}
if, from round $T$ on, exactly one node has its status set to $\elected$ 
while all other nodes are in state $\nonelected$. This is the requirement 
for standard (implicit) leader election.  

We also study a problem called implicit agreement, which is a generalization
of leader election and agreement, another fundamental problem \cite{AMP18}. 

\para{Implicit Agreement.}
Assume initially each node has an input value in $\{0, 1\}$. An {\em implicit agreement} holds when the nodes' final states are either all contained in $\{0, \bot\}$ or all contained in $\{1, \bot\}$, and at least one node has state other than $\bot$ (which should be the input value of some node), where $\bot$ denotes the `undecided' state. In other words, all the decided nodes must agree on the {\em same value which is an initial input value of some node} and there must be {\em at least one decided node} in the network.

\section{Non-Oblivious Quantum Distributed Computing}
\label{sec:qdc}
We now describe the model of non-oblivious distributed computing in the quantum setting, that
we introduce in this paper, and for which we provide quantum advantage.
Again, we consider a network of $n$ nodes, represented as an undirected connected graph $G=(V,E)$. 
Then, our model basically allows a node to select quantumly which nodes it communicate to --- that
is, its choice is in a quantum superposition --- instead of selecting via a random distribution as in the classical setting. Of course, the message itself may consists of quantum bits.

\subsection{Quantum Routing and Message Complexity}
At every time, the system's configuration can be in a superposition of every possible behavior of deterministic configurations. The transition from one configuration to another is done according to a unitary transformation, made of two steps: (1) Perform the same local unitary to each node on their local data, local memory, and reception/emission registers; (2) Send non-empty messages prepared in step (1), where, implicitly, a non-empty register selects a node to send a message.

The \emph{message complexity} $M$ of a distributed algorithm is simply the sum of the message complexities during each of its rounds.
Deterministically, a round of communication has \emph{message complexity} $C$ when the network carries at most $C$ messages of $O(\log(n))$ quantum bits in that round.
Quantumly, we extend this notion to a superposition of configurations.
A round of \emph{quantum} communication has \emph{message complexity} $M$ when the global state of the network 
is in a superposition of (deterministic) configurations with message complexity at most $M$.

A formal model for quantum routing, with an example of use, is given in \Cref{app:formal,app:example}. 

Finally, we point out that this paper constructs algorithms whose message complexity is a random variable due to initial random choices of the nodes, but whose round complexity is always bounded.

\subsection{Quantization of (Distributed) Algorithms}\label{sec:quantization}
Often, when using the quantum framework for Grover search, quantum amplification, and search via quantum walks, some inner procedures are described classically using deterministic or randomized algorithms.
This is for two reasons. First, such descriptions are more intuitive. Second, it is always possible to consider their quantum analogue with the same behavior, except that they are described by a unitary map that can be reversed. This is crucial when such a procedure needs to be boosted using one of the mentioned techniques above.

This approach is quite common in the literature of sequential quantum algorithms but may be less known in quantum distributed computing. We review and detail some of the arguments already presented in~\cite{LM18} in \Cref{app:quantization}. 

\begin{lemma}[Informal]
Let $A$ be a randomized or quantum distributed algorithm, possibly with intermediate measurements. Then, there is a quantum distributed procedure $B$, without any intermediate measurement, simulating $A$ with the same round and message complexities.
\end{lemma}

\section{Distributed Quantum Subroutines for Message Complexity}
\label{sec:subroutines}

We describe the main tools used to build our algorithms.
Since they are built upon sequential quantum algorithms, we first explain how we will distribute them. In particular, since we will allow decentralization, synchronization will be crucial. This justifies a particular attention on the predetermined upper bounds on our quantum procedures accordingly to the following definition.

\begin{definition}
Assuming that the network is initially synchronized,
we say that a distributed algorithm $A$ \emph{runs in $T$ rounds} when the computation of every node ends in at most $T$ rounds for every possible initial configuration.
\end{definition}

\subsection{Distributed Execution of a Sequential (Quantum) Algorithm}

The distributed execution of sequential algorithms is a standard technique for taking advantage of a distributed network. In the context of quantum computing, the sequential Grover search algorithm has been adapted to distributed computing~\cite{LM18} in a centralized way. More precisely, a designated node, say $v$, is in charge of simulating the sequential algorithm (over the distributed network) and whenever it needs some information from, or action to be taken by, other nodes in the network, $v$ contacts the appropriate nodes using standard distributed communication primitives: e.g., building a spanning tree and then using the upcast or downcast primitives, or when information can be aggregated, the more efficient convergecast and broadcast primitives. These standard primitives have an inherent round complexity at least proportional to the diameter of the graph (and possibly linear in the number of nodes), and a communication complexity usually proportional to the number of edges in the graph.

Since, in this work, we aim for \emph{sublinear} message complexity, we consider even more distributed scenarios. We replace the use of the above communication primitives with two types of subroutines, or a mixture of both:
(1) \emph{centralized}: when nodes collaborate with $v$ whenever they are requested to do so by $v$, using some (quantumly) truncated communication primitive;
(2) \emph{decentralized}: when nodes collaborate without being notified to do so, which is possible using a careful synchronization of the network.

Note that most of our algorithms will use only a centralized adaptation of some sequential quantum subroutines, except for our leader election algorithm in diameter-2 networks, which is based on search via quantum walks.

\subsection{Distributed Search and Counting: General Setting}\label{sec:searchsetting}
In the following, $X$ is a finite set,  $f\colon X\to \{0,1\}$, and $u$ is any fixed node of a network $G$. 
In addition, there is a distributed algorithm $\checking$ that enables $u$ to compute $f(x)$, for input $x$ known by $u$. 
For synchronization constraints, we will need an upper bound on the number of rounds used by $\checking$, for any input $x\in X$.

More formally, let $\ket{\psi}_G$  be the initial state of the network, which may result from some preliminary initialization steps, and that we assume without loss of generality to be a quantum state.
Given $x\in X$ in $u$, $\checking$ enables $u$ to compute $f(x)$ in time $T_\Cc$ with message complexity $M_\Cc$ as follows, where the subscript $u$ means that the register is local to $u$, and $G$ that it is global to the network:
$$\checking: \ket{x,0}_u \ket{\psi}_G\mapsto \ket{x,f(x)}_u\ket{\phi_x}_G.$$
As opposed to previous works such as~\cite{LM18}, the algorithm $\checking$ can be centralized or decentralized, as explained in the above section.

The purpose of the rest of this section is to provide quantum subroutines that either
find $x\in X$ such that $f(x)=1$ (assuming such elements exist),
or that estimate the number of such instances of $x$.
For that purpose, we will use parameters $t_f=|f^{-1}(1)|$ and $\eps_f=t_f/|X|$.

Last an example of application on the star graph of those routines is given in \Cref{ref:starexample}.

\subsection{Grover Search}

In this part, the purpose of $u$ is to find $x\in X$ such that $f(x)=1$ (assuming such elements exist).  

A classical strategy is to sample $x\in X$ and then check whether $f(x)=1$ using $\checking$. Then the success probability is $\eps_f$, which can be boosted to $(1-\alpha)$ with $\Theta( \log(\tfrac{1}{\alpha})/{\eps_f})$ iterations. 
When $\eps_f$ is not known but satisfies $\eps_f=0$ or $\eps_f\geq \eps$, where $0<\eps\leq 1$ is some input parameter, one can guarantee that after  $\Theta({\log(\frac{1}{\alpha})}/{\eps})$ iterations we can distinguish between the two cases with arbitrarily high success probability.
In the distributed setting, it just means that the round and message complexities of $\checking$ are multiplied by this number of iterations, since the sampling procedure is local to $u$.

But quantumly we can do quadratically better. The following is a distributed adaptation of the Grover Search algorithm~\cite{Grover96}, in the case of an unknown number of pre-images of $f$~\cite[Lemma 2]{bbht98}, to the distributed setting. 
Such an adaptation was firstly done in~\cite{LM18} for the round complexity.
Here we consider also the message complexity, and the possibility of $\checking$ to be decentralized as explained in the previous subsection.

\begin{theorem}[Distributed Grover Search]\label{dgs}
Let $f,u,\eps_f$, $\checking$, $T_\Cc$ and $M_\Cc$ be defined as in \Cref{sec:searchsetting}.

For any $\eps,\alpha>0$, there is a quantum distributed algorithm $\mathsf{GroverSearch}(\eps,\alpha)$
such that
\begin{enumerate}
\item $\mathsf{GroverSearch}(\eps,\alpha)$ runs in $O({\log(\frac{1}{\alpha})}\times\frac{T_\Cc}{\sqrt{\eps}})$ rounds with message complexity $O({\log(\frac{1}{\alpha})}\times\frac{M_\Cc}{\sqrt{\eps}})$;
\item $\mathsf{GroverSearch}(\eps,\alpha)$ returns to
$u$ some $x\in X$, which satisfies $f(x)=1$ with probability at least $1-\alpha$ when $\eps_f\geq \eps$.
\end{enumerate}
\end{theorem}
\begin{proof}
Grover Search consists consists in at most $\lfloor a \log (1/\alpha)\rfloor$ attempts, for some constant $a$, each of them being in fact the Grover algorithm. We now decompose the proof in 7 steps.

\emph{Overview of Grover algorithm:}
Grover Search consists in a random number $t$
of iterations of the Grover operator $R=D\times S_f$ on a starting state made of a uniform superposition over $X$. The unitary $R$ will operate as a rotation from the starting state toward a uniform solution of $f^{-1}(1)$. More precisely:
\begin{itemize}
\item Starting state: $\ket{s}=\tfrac{1}{\sqrt{|X|}}\sum_{x\in X} \ket{x}$;
\item Unitary map $D$: Reflection through the starting state $\ket{s}$;
\item Unitary map $S_f$: $\ket{x}\mapsto (-1)^{f(x)}\ket{x}$, which is also a reflection;
\item $t\leq \lfloor b \log (1/\alpha)\rfloor$, for some constant $b$.
\end{itemize}

\emph{Synchronization:}
Whereas the whole network knows $a,b,\eps,\alpha$, node $u$ cannot share either the number of attempts of Grover algorithm, or the number of iterations of $R$, to the whole network, since it would be too expensive in message complexity.
For the number of attempts, we just continue them until the limit is reached, even if $x$ such that $f(x)=1$ is found before. For the number of iterations, the network will also assume the worst possible value, and continue to apply $\checking$ in case of necessity. We will explain at the end of the proof how this can be done without affecting the global state of the network.

\emph{Global state of the network:} Before continuing, we remind that the global state of the network is $\ket{\psi}_G$, due to some possible initialization step (\Cref{sec:searchsetting}). So finally the initial state is $\ket{s}_u\ket{\psi}_G$. The reflections $D$ and $S_f$ are extended without taking into account the global state, even if the implementation of $S_f$ may use it as a catalyst. The crucial part is that is should come back to its initial state $\ket{\psi}_G$ so that it can be disregarded from the point of view of $u$, leading to a valid implementation of $R$.

\emph{Preparation of the starting state:}
Node $u$ creates the starting state $\ket{s}_u$ and is in charge of synchronizing the iterations of $R=D\times S_f$ to the starting state. 

\emph{Realization of $R$:}
Since $D$ is independent from $f$, its application is fully local, assuming that $u$ holds the state. For $S_f$, $u$ will take benefit of $\checking$ in order to decide whether the amplitude of $\ket{x}$ needs to be flipped. 

\emph{Realization of $S_f$:}
Assuming that $\checking$ has been quantized. It should realize a unitary acting on $x$ and an extra bit initially set to $0$ as following:
$$\checking : \ket{x,0}_u \ket{\psi}_G \mapsto \ket{x,f(x)}_u\ket{\phi_x}_G,$$
where we take into consideration that the state $\ket{\psi}_G$ of the network may change and depend on $x$ because of $\checking$.

Let us define the phase-flip unitary $\mathit{PF}:\ket{b}\mapsto (-1)^b\ket{b}$ on the bit encoding $f(x)$. 
Then 
$$PF\times\checking: \ket{x,0}_u \ket{\psi}_G \mapsto (-1)^{f(x)} \ket{x,f(x)}_u \ket{\phi_x}_G,$$
Observe that one could also realize $(\checking)^{-1}$, the inverse of $\checking$, which is well defined since $\checking$ is unitary. Indeed, given a sequence of operations realizing $\checking$, realizing its inverse can be done by doing the inverse of each operations in reverse order.
Therefore
$$ (\checking)^{-1}\times \mathit{PF} \times\checking: \ket{x,0}_u \ket{\psi}_G \mapsto (-1)^{f(x)} \ket{x,0}_u \ket{\psi}_G.$$

\emph{Iterating and ending:}
As we have explained above, $u$ can prepare the starting state $\ket{s}_u$, run $D$, then the network simulates $S_f$, then $u$ run $D$ again and so on. Finally, the process stops locally when $u$ decides to stop the number of alternations of $D$ and $S_f$, whereas the network still continues to assist $u$.
More precisely, the network goes through the steps of $S_f$ while $u$ simply neither does $D$ nor $\mathit{PF}$. In which case, the network transformation is the identity.
\end{proof}

\subsection{Quantum Counting}
It is also possible to have an extension of Grover search for speeding up the statistic approximate counting based on sampling. We proceed in two steps. The first one is an adaptation of~\cite{BHT98c} from the sequential setting to the distributed one.
\begin{theorem}[Distributed Quantum Counting]
Let $f,u,t_f$, $\checking$, $T_\Cc$ and $M_\Cc$ be defined as in \Cref{sec:searchsetting}.
Then there is a quantum algorithm $\mathsf{Count}(P)$  such that
\begin{enumerate}
\item $\mathsf{Count}(P)$ runs in $P\times T_\Cc$ rounds with message complexity $P\times M_\Cc$;
\item When  $P\geq 4$ and $t_f\leq |X|/2$: $u$ outputs $\tilde{t}_f$ which satisfies
$|t_f - \tilde{t}_f| < \frac{2 \pi}{P} \sqrt{t_f|X|} + \frac{\pi^2}{P^2} |X|$ with probability at least  $8/\pi^2$.
\end{enumerate}
\end{theorem}
\begin{proof}
We simulate the sequential approach of \cite{BHT98c}, which uses as a subroutine the Grover operator $R$, see proof of \Cref{dgs}, together with a procedure called Phase Estimation.
Phase Estimation runs several iterations of $R$, chosen in superposition, on the same starting state as in the proof of~\Cref{dgs}.

Only the implementation of $R$ will be done using the network, and the rest of Phase Estimation will be done locally by $u$. Then $R$ can be iterated a number of times solely controlled by $u$, even using the network for the procedure $\checking$, required to implement $S_f$. 
This part is similar to the proof of~\Cref{dgs}.
\end{proof}

We now state a useful direct corollary.
\begin{corollary}
\label{cor:approxCount}
Let $f,u,t_f$, $\checking$, $T_\Cc$ and $M_\Cc$ be defined as in \Cref{sec:searchsetting}.
For any $\alpha,c>0$, there is a quantum algorithm $\mathsf{ApproxCount}(c,\alpha)$ such that
\begin{enumerate}
\item $\mathsf{ApproxCount}(c,\alpha)$ runs in $O({\log(\frac{1}{\alpha})}\times \frac{T_\Cc}{c})$ round with message complexities $O({\log(\frac{1}{\alpha})}\times \frac{M_\Cc}{c})$;
\item $u$ outputs $\tilde{t}_f$ which satisfies
$|t_f - \tilde{t}_f| <  c|X|$ with probability at least   $1-\alpha$.
\end{enumerate}
\end{corollary}
\begin{proof}
When $t_f\leq |X|/2$, and $P=4\pi/c$, the approximation error is at most
$c/2\sqrt{t_f|X|}/2+c^2 |X|/16\leq c |X|$, with probability at least $8/\pi^2$.
In order to boost the probability to $1-\alpha$, one only need to compute the median of the  outputs of $\log(\frac{1}{\alpha})$ runs of that procedure.

In the general case, we don't know whether $t_f\leq |X|/2$ or not. Nonetheless we can transform the function so that this hypothesis holds. To simplify the discussion, assume that $X=[N]$. Then, we simply consider another function $g:[2N]\to\{0,1\}$ which coincides with $f$ on $[N]$ and takes the value otherwise. This function also satisfies $t_g=|g^{-1}(1)|=t_f$, but now $t_g\leq N=(2N)/2$. So we can apply the previous approach with $P=8\pi/c$ instead.
\end{proof}

\subsection{Distributed Search via Quantum Walks}

We revisit distributed Grover search using quantum walks as in the framework of~\cite{mnrs}, that we adapt to distributed computing.
We consider an irreducible and reversible Markov chain $P=(p_{xy})_{x,y\in X}$ on $X$, with stationary distribution $\pi$ and eigenvalue gap $\delta$. In this context we now define $\eps_f=\Pr_{x\sim \pi}(f(x)=1))$.

We also have two new procedures $\setup$ and $\update$, which prepares the network before using $\checking$. In the sequential setting, they are used to maintain a database used by $\checking$. For us, this is like a distributed database. Assuming that initially the network is in state $\ket{\psi}_G$, they act as follows:
\begin{itemize}
\item $\setup$: Given $x\in X$ in $u$, it maps the current network state $\ket{\psi}_G$ to a new state $\ket{\phi_x}_G$ that can depend on $x$ in time $T_\Sc$ with message complexity $M_\Sc$:
$$\setup: \ket{x}_u \ket{\psi}_G\mapsto \ket{x}_u\ket{\phi_x}_G;$$
\item $\update$: Given $x,y\in X$ such that $p_{xy}\neq 0$, it  updates the network state for $x$ to one for $y$ in time $T_\Uc$ with message complexity $M_\Uc$:
$$\update: \ket{x,y}_u \ket{\phi_x}_G\mapsto \ket{x,y}_u\ket{\phi_y}_G;$$
\item $\checking$: Given $x\in X$ and network state $\ket{\phi_x}_G$,
$\checking$ computes $f(x)$ in time $T_\Cc$ with message complexity $M_\Cc$:
$$\checking: \ket{x,0}_u \ket{\phi_x}_G\mapsto \ket{x,f(x)}_u\ket{\phi_x'}_G.$$
\end{itemize}

\begin{theorem}[Distributed search via quantum walk]\label{qwalk}
Let $f,u,P,\pi,\delta,\eps_f$, $\setup,\update,\checking$, $T_\Sc,T_\Uc,T_\Cc$, and $M_\Sc,M_\Uc,M_\Cc$,  be defined as above.
For any $\eps,\alpha>0$, there is a quantum distributed algorithm 
$\mathsf{WalkSearch}(P,\delta,\eps,\alpha)$ such that
\begin{enumerate}
    \item $\mathsf{WalkSearch}(P,\delta,\eps,\alpha)$ runs in 
    $O\left(\log(\frac{1}{\alpha}) \times \left(T_\Sc+\frac{1}{\sqrt{\eps}}\left( \frac{1}{\sqrt{\delta}}T_\Uc+T_\Cc\right)\right)\right)$
    rounds and
with message complexity
    $O\left(\log(\frac{1}{\alpha}) \times \left(M_\Sc+\frac{1}{\sqrt{\eps}}\left( \frac{1}{\sqrt{\delta}}M_\Uc+M_\Cc\right)\right)\right)$;
\item $\mathsf{WalkSearch}(P,\delta,\eps,\alpha)$ returns to
$u$ some $x$ in the support of $\pi$, which satisfies $f(x)=1$ with probability at least $1-\alpha$ when $\eps_f\geq \eps$.
\end{enumerate}
\end{theorem}
\begin{proof}
Let us start similarly to the proof of \Cref{dgs}, where now
\begin{itemize}
\item Starting state in $u$ while updating the network:  $\ket{s}=\sum_{x\in X} \sqrt{p_x}\ket{x,0}_u\ket{\phi_x}_G$;
\item Unitary map $D$: Reflection through the starting state $\ket{s}$;
\item Unitary map $S_f$: $\ket{x}_u\ket{\phi_x}_G\mapsto (-1)^{f(x)}\ket{x}_u\ket{\phi_x}_G$.
\end{itemize}
First it is clear that the starting state can be prepared using one use of $\setup$, and the unitary map 
$S_f$ performed using one use of $\checking$.
The rest of the proof is devoted to the realization of $D$. It relies on the use of a quantum walk to realize $D$ as in~\cite{mnrs}. The core of the proof is to realize the local reflection while updating the network state.

Fix for now $x\in X$, and let us focus on the state in $u$. 
Define $\ket{p_x}=\sum_{y\in X}\sqrt{p_{xy}}\ket{y}$. Let us call $D_x$ be the reflection through $\ket{p_x}$, and $A$ be that reflection controlled on $x$:
$$A=\sum_x \ketbra{x}\otimes D_x : \sum_{xy}\alpha_{x,y}\ket{x,y}\mapsto \sum_{xy}\alpha_{xy}\ket{x} \otimes (D_x\ket{y}).$$
Then the quantum walk operator is defined as
$$W(P)= A \times \mathrm{SWAP} \times A.$$
This is the core operator to implement $D$ using Phase Estimation with $\Theta(\frac{1}{\sqrt{\delta}})$ uses of $W(P)$. We are going to skip that part and refer to~\cite{mnrs}. 

The only thing that remains to prove is to maintain the network state compatible with the first register using $\update$. The application of $A$, does not affect the first register, so there is nothing to do.
We can now focus on the $\mathrm{SWAP}$ operation. Given the state $\ket{x,y}_u\ket{\phi_x}_G$ we would like to produce $\ket{y,x}_u\ket{\phi_y}_G$. This is exactly done by applying first $\update$ then $\mathrm{SWAP}$.

Putting all the pieces together gives the result.
\end{proof}

\section{Quantum Leader Election}
\label{sec:leader}

In this section, we give leader election algorithms, over several network configurations, and without prior shared randomness and entanglement. These leader election algorithms have (quantum) message complexities that are significantly better than the best known message complexities (on the respective network configurations). Moreover, for all but the leader election on graphs with a specified mixing time (cf. Section \ref{sec:LEonExpander}), the resulting quantum message complexity goes significantly below the corresponding classical message complexity lower bounds.

\subsection{Leader Election in Complete Networks}
\label{sec:implicitLE}\label{sec:leader-complete}

We describe a quantum protocol for implicit leader election (see \Cref{sec:pbs}). The main challenge is to beat, in the distributed setting, the birthday paradox on which the protocol of~\cite{KPPRT15} is based,
just as this was done by a sequential quantum algorithm for finding collisions in a random function~\cite{BHT98}.
In order to apply this technique, we need to break the symmetry of the original protocol.

\subsubsection{Algorithm Description.} The leader election protocol \quantumLE{} (\Cref{alg:leaderComplete}) is separated into two phases: a classical phase, followed by a quantum phase. It is also parametrized by some integer $k \in [1,n^{1/3}]$, which allows for a trade off between rounds and messages. Of particular interest is when $k = \Theta(n^{1/3})$, in which case the resulting protocol has optimal message complexity $\tilde{\Theta}(n^{1/3})$. 

In the classical phase, each node becomes a \emph{candidate} with probability $p = (\pCst \ln n)/ n$. Then, each candidate node first generates a (uniformly) random rank $r_v \in \{1,\ldots,n^4\}$ and sends that rank to $k$ neighbors chosen arbitrarily. When the classical phase terminates, we can define for each candidate node $v$ a (global) function $f_v : V \rightarrow \{0,1\}$ that assigns for any node $w \in V$ the value 1 if and only if $w$ received a rank strictly higher than that of $v$ (in the classical phase). 

In the quantum phase, candidate nodes determine whether they hold the highest rank or not (among candidate nodes), using $\tilde{O}(\sqrt{n/k})$ rounds and messages. To do so, each candidate node $v$ executes $\mathsf{GroverSearch}(\eps,\alpha)$ (\Cref{dgs}) with $\eps = k/n$, $\alpha = 1/n^2$, to search for nodes in $f_v^{-1}(1)$, or in other words, for any node that received a higher rank than that of $v$ in the classical phase. If, as a result, candidate node $v$ finds no node $w \in V$ such that $f_v(w)=1$, then $v$ becomes the leader. After these $\tilde{O}(\sqrt{n/k})$ rounds, all nodes terminate. 

When $v$ executes $\mathsf{GroverSearch}(\eps,\alpha)$, it uses a simple distributed algorithm $\checking_v$ to compute for any node $w \in V$ the value of $f_v(w)$ using two rounds and two messages; that algorithm consists of $v$ sending a message with its rank to $w$ in a first round, and receiving a reply with $f_v(w)$ in the second round. 

\begin{algorithm}[ht]\small
\begin{algorithmic}[1]
\Require A complete $n$-node anonymous network.
\Ensure Leader Election. \smallskip

\item[\textbf{Choosing Candidates (Classical):}]
  
  \State  Every node $v$ decides to become a candidate with probability $\frac{12\ln n}{n}$
  and generates a random rank $r_v$ from $\{1,\dots,n^4\}$.

   \State   If a node $v$ does not become a candidate, then it
  immediately enters the $\nonelected$ state;
  otherwise, it executes the next step.

 \item[\textbf{Choosing Referees (Classical):}]
  \State 
  Each candidate node $v$ contacts an arbitrary set of $k$  nodes  and sends its rank to all the $k$ nodes.\smallskip

  \item[\textbf{Distributed Grover Search (Quantum):}] 
  \State  Each candidate node $v$ runs $\mathsf{GroverSearch}(k/n,1/n^2)$ to search for some node $w \in V$ having received a rank strictly higher than that of $v$ in the classical phase (i.e., with $f_v(w)=1$).
\smallskip

 \item[\textbf{Decision (Classical):}] 
    \State 
If a candidate node $v$ finds no node $w \in V$ such that $f_v(w)=1$, then $v$ enters the $\elected$ state (becomes the leader). Otherwise, it enters the $\nonelected$ state.
\end{algorithmic}
\caption{(\quantumLE{}) {Quantum leader election protocol for complete networks}}
\label{alg:leaderComplete}
\end{algorithm}

\subsubsection{Analysis.} First, we recall that \Cref{obs:sampleAndUniqueRanks} in Subsection \ref{subsec:randomizedPrelims} captures well-known statements regarding sampling and choosing unique ranks, and we will use these in the following analysis. More concretely, with probability at least $1-1/n^2$, it holds both that the number of candidate nodes is non-zero but at most $\pMax \ln n$, and that these candidates choose unique random ranks. 

Next, we give a key observation on the (global) functions $f_v$, defined for any candidate node $v$, which implies that our use of (distributed) Grover Search is correct. Recall that these are defined as $f_v : V \rightarrow \{0,1\}$ such that for any node $w \in V$,  $f_v(w) = 1$ if and only if $w$ received a rank strictly higher than that of $v$ in the classical phase.

\begin{fact}
\label{obs:electionFunctions}
    With probability at least $1-1/n^2$, it holds both that:
    \begin{itemize}
        \item There exists a single candidate node $l$ such that $|f_l^{-1}(1)| = 0$, and it is the candidate node with the highest rank. 
        \item For any other candidate node $u \neq l$, $|f_u^{-1}(1)| \geq k$. 
    \end{itemize}
\end{fact}
\begin{proof}
    By \Cref{obs:sampleAndUniqueRanks}, with probability at least $1-1/n^2$, at least one, and at most $\pMax \ln n$, nodes become candidates, and they all choose unique ranks. Let $l$ be the candidate node with the highest rank. Then, no node in $V$ receives any rank strictly higher than that of $l$, and thus $|f_l^{-1}(1)| = 0$. Moreover, in the classical part, node $l$ sends its rank to at least $k$ other nodes. Hence, for all other candidate nodes $u \neq l$, $|f_u^{-1}(1)| \geq k$. 
\end{proof}

Now, we can prove that the above algorithm solves implicit leader election, and in fact does so significantly more message-efficiently than any classical algorithm can. First, we give our round and message complexity upper bounds parameterized by some integer $k \geq 1$. 
\begin{theorem}\label{thm:quantumle}
    \quantumLE{} solves implicit leader election with probability at least $1-1/n$.  Moreover, it takes $\tilde{O}(\sqrt{n/k})$ rounds and with probability at least $1-1/n$, it sends $\tilde{O}(k+ \sqrt{n/k})$ messages.
\end{theorem}

\begin{proof}
    We first show correctness. By \Cref{obs:electionFunctions}, with probability at least $1-1/n^2$, there exists a single candidate node $l$ such that $|f_l^{-1}(1)| = 0$, which is the candidate node with the highest rank, whereas for any other candidate node $u \neq l$, $|f_u^{-1}(1)| \geq k$. First, since $|f_l^{-1}(1)| = 0$, then by Theorem \ref{dgs}, $l$ declares with probability 1 that no element $w \in V$ with $f_l(w) = 1$ was found, and becomes leader. On the other hand, for any other candidate node $u \neq l$, we have $\eps_f = |f_u^{-1}(1)|/n \geq k/n$. As $\eps_f \geq \eps$, by Theorem \ref{dgs}, candidate $u$ outputs some element $w \in V$ with $f_u(w) = 1$ with probability $1-\alpha=1-1/n^2$, and in which case $u$ does not become leader. Thus, by a union bound over all (the at most $n$) candidate nodes $u \neq l$, the following statement holds with probability at least $1- 1/n$: the candidate node with highest rank becomes a leader, and it is the only node to do so.

    Next, we prove the round and message complexities. First, by \Cref{obs:sampleAndUniqueRanks}, we upper bound the number of candidates by $O(\log n)$, with probability at least $1-1/n^2$. We condition the remainder of the proof of the event that there are at most $O(\log n)$ candidates. Since each candidate sends $k$ messages in a single round, the classical phase takes $O(1)$ rounds and uses $O(k \log n)$ messages. As for the quantum phase, by Theorem \ref{dgs} and the fact that for any candidate node $v$, $\checking_v$ takes 2 rounds and messages, we get that each $\mathsf{GroverSearch}(k/n,1/n^2)$ 
    sends $\tilde{O}(\sqrt{n/k})$ messages over $\tilde{O}(\sqrt{n/k})$ rounds. Note that the calls of some candidate node $v$ during both the classical and quantum phases use edges different from those used by calls from any other candidate node $u$. Hence, the number of candidate nodes induces no (edge) congestion and thus runtime overhead, but only a $O(\log n)$ message overhead (conditioned on having $O(\log n)$ candidates). It follows that the overall round complexity is $\tilde{O}(\sqrt{n/k})$ rounds deterministically, and the overall message complexity is $\tilde{O}(k + \sqrt{n/k})$ messages with probability at least $1-1/n^2$.
\end{proof}

Setting $k = \Theta(n^{1/3})$ optimizes the message complexity of the above leader election protocol. The resulting message complexity is significantly better than the \emph{tight} $\tilde{\Theta}(\sqrt{n})$ message complexity of~\cite{KPPRT15}.

\begin{corollary}
    \quantumLE{} solves implicit leader election with probability at least $1-1/n$.  Moreover, it takes $\tilde{O}(n^{1/3})$ rounds and with probability at least $1-1/n$, it sends $\tilde{O}(n^{1/3})$ messages.
\end{corollary}

However, the above improvement to message complexity for leader election in complete networks comes at the cost of a significantly increased runtime. Still, we point out that even if the aim is $o(n^{1/3})$ runtime, \quantumLE{} can obtain a message complexity that goes below the classical setting's best achievable message complexity of $\tilde{\Omega}(\sqrt{n})$. Indeed, when $k = n^{5/12}$, \quantumLE{} takes $\tilde{O}(n^{7/24}) = o(n^{1/3})$ rounds, and uses only $\tilde{O}(n^{5/12}) = o(\sqrt{n})$ messages (with high probability).

\subsection{Leader Election in Graphs with Mixing Time \texorpdfstring{$\tau$}{tau}}
\label{sec:LEonExpander}

Now, we show how to extend the previous algorithm \quantumLE{} to any (communication) network $G$ by using random walks on $G$.
This will result in an efficient algorithm when $G$ is an expander, or when its mixing time $\tau$ is small enough. 
The main idea is to replace an exploration on the neighborhood of a vertex, by a random walk from this vertex.
Nonetheless, there is a technical subtlety: due to the centralization of one part of Grover search, we cannot just walk on the network.

\subsubsection{Algorithm Description.} The leader election protocol \quantumRWLE{} (\Cref{alg:leaderexpander}) for general graphs
with given mixing time $\tau$ generalizes the one for complete graphs (\Cref{alg:leaderComplete}).
The exploration of the neighborhood of a candidate vertex $v$ is replaced by a \emph{random walk on the network}.
The idea of using a random walk is taken from the classical protocol in~\cite{KPPRT15},
with which they achieve a message complexity of $\tilde{O}(\tau \sqrt{n})$.
Due to the centralization of one part of $\mathsf{GroverSearch}$, we will need $v$ to control the random choices of the walk, leading to a blow up of $\tau$ in the message complexity.

\emph{This random walk is not used for searching via quantum walk.} This is just an inner procedure for the initialization and also for implemented $\checking$.

The protocol is again separated into two similar phases. In the classical phase, each candidate node generates a $k$ random walk tokens containing its rank $r_v$. Finally, all of the candidates' random walk tokens take $O(\tau)$ steps (over the communication network $G$). In this step, the random choices are delegated to the corresponding nodes of the walk. So the message and round complexity are simply $\tilde{O}(\tau k)$.

In the quantum phase, each candidate node $v$ queries an $\Theta(\tau)$-length random walk starting at $v$ (over the communication network $G$), and the random walk is said to be good if it ends at a node having received a higher rank (than $v$'s rank). 
Due to the centralization of one part of $\mathsf{GroverSearch}$, we ask $v$ to \emph{make} the random choices of the walk and to \emph{propagate} them through the walk itself, leading to a blow up of the message (and round) complexity from $\tilde{O}(\tau)$ to $\tilde{O}(\tau^2)$ for each walk.
In other words, node $v$ runs $\mathsf{GroverSearch}$ with $\eps=k/n$, $\alpha = 1/n^2$, and with the function $f_v: X \rightarrow \{0,1\}$ that assigns to any $\Theta(\tau)$-length random walk (made of $O(\tau\log n)$ random bits) in $X$ the value 1 if the walk ends at a node having received some rank $r_w > r_v$ in the classical phase, and 0 otherwise.
If as a result, the candidate node $v$ does not find any good random walks, then $v$ becomes the leader.

\begin{algorithm}[ht]\small
\begin{algorithmic}[1]
\Require A $n$-node anonymous network with mixing time $\tau$
\Ensure Leader Election. 
\smallskip

\item[\textbf{Choosing Candidates (Classical):}] See \Cref{alg:leaderComplete}
 
 \item[\textbf{Choosing Referees (Classical):}]
  \State 
  Each candidate node $v$ contacts and sends its rank to an arbitrary set of $k$ nodes  through $k$ random walks of length $\Theta(\tau)$.
\smallskip

  \item[\textbf{Distributed Grover Search (Quantum):}] 
  \State  Each candidate node $v$ runs $\mathsf{GroverSearch}(k/n,1/n^2)$ to search for some $\Theta(\tau)$-length random walk $r$ reaching a node $w \in V$ having received a rank strictly higher than that of $v$ in the classical phase (i.e., with $f_v(r)=1$).
 \smallskip
 
 \item[\textbf{Decision (Classical):}] 
    \State 
If a candidate node $v$ finds no random walks $r$ such that $f_v(r)=1$, then $v$ enters the $\elected$ state (becomes the leader). Otherwise, it enters the $\nonelected$ state.
\end{algorithmic}
\caption{(\quantumRWLE{}) {Quantum leader election protocol for graphs with mixing time $\tau$}}
\label{alg:leaderexpander}
\end{algorithm}

\subsubsection{Analysis.}

\begin{theorem}
    \quantumRWLE{} solves implicit leader election with probability at least $1-1/n$.  Moreover, it takes $\tilde{O}(\tau k+  \tau^2\sqrt{n/k})$ rounds and with probability at least $1-1/n$, it sends $\tilde{O}(\tau k+ \tau^2\sqrt{n/k}))$ messages.
\end{theorem}

\begin{proof}
The proof is similar to \Cref{thm:quantumle}, where explorations are replaced by random walks, and the function evaluations $f_u(w)$ for node $w$ adjacent to $u$ is replaced by the evaluation of $f_u(r)$ for any $\Theta(\tau)$-length random walk $r$ starting at $u$.
\end{proof}

Setting $k = \Theta(\tau^{2/3}n^{1/3})$ optimizes the message complexity of the above leader election protocol. 

\begin{corollary}
    \quantumRWLE{} solves implicit leader election with probability at least $1-1/n$.  Moreover, it takes $\tilde{O}(\tau^{5/3}n^{1/3})$ rounds and with probability at least $1-1/n$, it sends $\tilde{O}(\tau^{5/3}n^{1/3})$ messages.
\end{corollary}
 
\subsection{Leader Election in Graphs with Diameter 2}
\label{sec:implicitLED}

We now turn to graphs with diameter $2$, probably our most challenging algorithm.
The algorithm will now use a quantum walk. This time the walk is not on the (communication) network as in \quantumRWLE{}.
Instead, the walk is on a data structure maintained locally by some nodes in order to speedup their Grover Search using \textsc{WalkSearch}. A similar quantum walk was used to solve element distinctness~\cite{Amb07}, which is basically the worst case of collision finding. Nonetheless, the situation here is much more complicated,
and in addition, one of our instances of Grover search is decentralized.

\subsubsection{Algorithm Description.}
The leader election protocol \quantumQWLE{} (\Cref{alg:leaderdiam2}) for graphs with diameter $2$ has similarities with protocol \quantumLE{} in \Cref{sec:implicitLE}. The main difference is that now the selection of the referees is done in superposition. Then a quantum walk is used in order to search for a subset of referees contradicting the leadership of a candidate.
For technical reasons, we will be able to test candidates only one by one. To break the symmetry, a candidate decides randomly to be either active or passive. Most of the time, no candidate node is active, so none of them is disregarded. With some probability of magnitude $\Theta(1/\log^2 n)$, a single node is active. Then, it is really tested. Otherwise, when more candidates are active, no guarantee can be given except that the candidate with the higher rank will never enter a $\nonelected$ state.

More precisely, active nodes challenge themselves against passive ones.
To do so, active nodes $v$ perform a distributed quantum search via quantum walk, that is the quantum routine
$\mathsf{WalkSearch}(P_v,\delta,\eps,\alpha)$ (\Cref{thm:quantumle}), where $\alpha=1/n^2$, and the  other parameters are described below.

For an active candidate node $v$, the quantum walk is based on a uniform random walk $P_v$ on the Johnson graph $J(\deg(v),k)$, whose spectral gap $\delta\approx 1/k$, when $k=o(n)$. 
We remind that the Johnson Graph $J(\deg(v),k)$ is a graph whose vertices are all subsets $W$ of size $k$ from a given universe of size $\deg(v)$. In our case, this universe is the set of nodes connected to $v$. Then two subsets are adjacent in $J(\deg(v),k)$ when they differ by exactly $1$ element. Since $P$ is uniform, the transition probabilities are uniform and the stationary distribution $\pi$ is also uniform over all possible subsets of size $k$.
Therefore setting $\eps=k/\deg(v)$ will suit the primitive's requirement, since this lower bounds the probability that a given referee belongs to a random subset $W$.

Notice that the walk $P_v$ is performed locally in $v$ on a \emph{virtual} graph different from the network graph. Nonetheless, $P_v$ guides the distributed quantum search over the network. Indeed, procedure $\setup$ propagates the rank $r_v$ of $v$ to each nodes of $W$. Procedure $\update$ adjusts this information to correspond to a new subset $W'$ obtained from $W$ by removing one element and adding a new one.
Finally, the purpose of $\checking$ is to implement the function $f$ such that $f(W)=1$ when there is a referee $w\in W$ connected to a passive candidate with higher rank.

Whereas $\setup$ and $\update$ are the direct quantization of their deterministic description,
$\checking$ will be implemented in $2$ phases. In the first one, a distributed Grover Search is done by each passive candidate node $v'$, in order identify any potential referee $w$ holding a smaller rank than $r_{v'}$. This search is decentralized since candidate nodes perform it on a given round even without being notified directly by an active candidate node. 
Then, another Grover Search is done by each candidate node $v$ in order to detect if such a referee has been informed of any passive active node with higher rank.

\begin{algorithm}[ht]\small
\caption{(\quantumQWLE) Quantum leader election protocol for graphs with diameter $2$}
\label{alg:leaderdiam2}
\begin{algorithmic}[1]
\Require A diameter $2$ $n$-node anonymous network.
\Ensure Leader Election.
\smallskip

\item[\textbf{Choosing Candidates (Classical)}:] See \Cref{alg:leaderComplete}
\smallskip

\item[\textbf{Iteratively Choosing and Testing Active Candidates (Classical \& Quantum)}:]
\For{$\Theta(\log^3 n)$ iterations}
\State Each candidate becomes active with probability $\Theta(1/\log^2 n)$, and passive otherwise \smallskip
\ForAll{active candidates $v$} \textbf{Distributed Search via Quantum Walk:}
\State Run \textsc{WalkSearch}$(P,k/\deg(v),1/k,1/n^2)$ with
\item[\hspace{2.5em}]  $P$: random walk on Johnson Graph $J(\deg(v),k)$ s.t. vertices are subsets $W$  made of $k$ neighbors of $v$
\item[\hspace{2.5em}]  $f$: function s.t. $f(W)=1$ iff there is $w\in W$ connected to a passive candidate $v'$ with higher rank
\State $\setup(W)$: For a subset $W$ of $k$ neighbors, send rank $r_v$ to all $w\in W$  \Comment{\emph{Contact Referees}}
\State $\update(W,W')$: For $W'=(W\setminus{\set{w})\cup\set{w'}}$, ask $w$ to send back $r_v$, and send $r_v$ to $w'$ \Comment{\emph{Replace $w$}}
\State $\checking(W)$: In two steps \Comment{\emph{Compute $f_v(W)$}}
\State {Decentralized step:} Each passive candidate $v'$ runs \textsc{GroverSearch}$(1/\deg(v'),1/n^3)$ to find any node $w$ in their own neighborhood, with a smaller rank (than that of $v'$) and sends the rank to $w$
\State {Centralized step:} Each active candidate $v$ runs \textsc{GroverSearch}$(1/k,1/n^3)$ to identify a node $w \in W$ which received a higher rank (than that of $v$) 
\EndFor
\smallskip

\item[\qquad \textbf{Decision}:]
\State Each active candidate node $v$ performs a last call to the above $\checking$ on their current set $W$
\If{a referee $w\in W$ is found which received a higher rank (than that of $v$)}
    \State Node $v$ decides to enter the $\nonelected$ state (and not participate to the next loop-iteration)
\Else
\State Node $v$ decides to remain a candidate 
\EndIf
\EndFor
\smallskip

\item[\textbf{Ending (Classical)}:]
\State Each remaining active candidate enters $\elected$ state
\end{algorithmic}
\end{algorithm}

\subsubsection{Analysis.}
\begin{theorem}
    \quantumQWLE{} solves implicit leader election with probability at least $1-1/n$.  Moreover, it takes $\tilde{O}(n/\sqrt{k})$ rounds and with probability at least $1-1/n$, it sends $\tilde{O}(k+ n/\sqrt{k})$ messages.
\end{theorem}
\begin{proof}
First, a direct inspection reveals that the candidate with the highest rank can never enter a $\nonelected$ state, since no referee will be able to collect a higher rank to contradict its election.

Second, if a unique candidate is active with a rank smaller than another (passive) candidate, then it will enter a $\nonelected$ state with high probability. The error probability is due to both \textsc{WalkSearch} in \quantumQWLE{} and the two instances of \textsc{GroverSearch} in $\checking$. Note that in \Cref{qwalk}, the procedure $\checking$ is supposed to be without error. However, since the number of its executions in \textsc{WalkSearch} is sublinear
and its error could be set to in $O(1/n^3)$, the overall error can be bounded by $O(1/n^2)$ since it accumulates linearly with the number of executions of $\checking$ (see, e.g., \cite[Box 4.1]{NC10}).

Last, we analyze the round and message complexities. 
We have $M_\Sc=\tilde{O}(k)$, $M_\Uc=\tilde{O}(1)$, $M_\Cc=\tilde{O}(\sqrt{n})$,
$T_\Sc=O(1)$, $T_\Uc=O(1)$, $T_\Cc=\tilde{O}(\sqrt{n})$,
and finally
$\eps=\frac{k}{n}$, $\delta=\frac{1}{k}$. Thus the final message complexity is
$$\tilde{O}\left(k+\sqrt{\frac{n}{k}}\big(\sqrt{k}+\sqrt{n}\big)\right)
=\tilde{O}\left(k+\sqrt{n}+\frac{n}{\sqrt{k}}\right)=\tilde{O}\left(k+\frac{n}{\sqrt{k}}\right).$$
\end{proof}

Setting $k = \Theta(n^{2/3})$ optimizes the message complexity of the above leader election protocol. 

\begin{corollary}
    \quantumQWLE{} solves implicit leader election with probability at least $1-1/n$.  Moreover, it takes $\tilde{O}(n^{2/3})$ rounds and with probability at least $1-1/n$, it sends $\tilde{O}(n^{2/3})$ messages.
\end{corollary}

\subsection{Leader Election in General Graphs via Tree Merging}
\label{sec:leader-general}

Finally, we describe an (explicit) leader election algorithm for general graphs with $\tilde{O}(n)$ round complexity and $\tilde{O}(\sqrt{m n})$ message complexity. Our algorithm solves leader election by cluster merging in a fashion similar to the well-known GHS algorithm~\cite{GHS}. The key difference lies in the use of Grover search to decide which clusters should be merged, far more message-efficiently (quadratically so, in fact) than can be achieved in the classical setting. More concretely, we use Grover Search to find a cluster's outgoing edges message-efficiently. On another note, our presented algorithm generalizes straightforwardly to the minimum spanning tree (MST) problem with the same complexities, which may be of independent interest.

\subsubsection{Algorithm Description.} The algorithm \quantumGeneralLE{} works in $O(\log n)$ phases, of $O(n \log n)$ rounds each. Initially, each node is its own cluster. Then, in any phase $i \geq 1$, we merge clusters together, in a way that ensures that whenever the phase starts with multiple clusters, the phase ends with at most half of the clusters. Let $\mathcal{C}_i$ be the collection of all clusters at the start of any phase $i \geq 1$. Each phase executes the following three steps:
\begin{enumerate}
    \item Each cluster $C \in \mathcal{C}_i$ computes some outgoing edge (i.e., with one endpoint in $C$ and the other outside). This happens in two parts. First, any node $v \in C$ searches for some incident edge (if there exists one) leading outside $C$ by running $\mathsf{GroverSearch}$ with $\eps=\sqrt{deg(v)}$, $\alpha = 1/n^2$ and the function $f_v: N(V) \rightarrow \{0,1\}$ that assigns $f(w) = 0$ to any node $w \in C$, and $f(w) = 1$ to any other node in $N(V)$. Then, in the second part, any node that finds such an incident, outgoing edge convergecasts it over the cluster tree, and if multiple nodes find an outgoing edge, the convergecast transmits any arbitrary one up to the cluster center.
    \item Clusters of $\mathcal{C}_i$ simulate a maximal matching algorithm on the \emph{virtual fragment (super)graph} $\mathcal{V}_i$ whose (super)nodes are the clusters of $\mathcal{C}_i$, and the edges correspond to cluster connected by the computed outgoing edges. This can be done using the Cole-Vishkin symmetry-breaking algorithm \cite{CV86}.
    \item Clusters of $\mathcal{C}_i$ merge together using the computed maximal matching, that is, (adjacent) matched supernodes (or clusters) merge together, whereas any unmatched supernode merges with some arbitrary neighboring matched node. The resulting clustering is $\mathcal{C}_{i+1}$.
\end{enumerate}
 Finally, after all $O(\log n)$ phases, a single cluster remains. The center of that cluster becomes the leader, and broadcasts its ID to all nodes via the cluster tree.
 
\subsubsection{Analysis.} First, we prove that in this quantum variant, step (1) of each phase uses significantly fewer messages (than in the classical setting) to find an outgoing edge for each cluster (if there exists one) through the use of Grover Search. 

\begin{lemma}
    Step (1) takes $O(n \log n)$ rounds and sends $O(\sqrt{m n} \log n)$ messages. Moreover, it guarantees that (the center of) any cluster $C \in \mathcal{C}_i$ finds an outgoing edge with probability at least $1-1/n^2$.  
\end{lemma}

\begin{proof}
    Consider the first part, in which each node $v \in V$ runs Grover Search with $\eps= 1 / deg(v)$, $\alpha = 1/n^3$ and the function $f_v: N(V) \rightarrow \{0,1\}$ that assigns $f(w) = 0$ to any node $w \in C$, and $f(w) = 1$ to any other node in $N(V)$. This function is evaluated by simple distributed algorithm $\checking_v$ to compute for any node $w \in V$ the value of $f_v(w)$ using two rounds and two messages; that algorithm consists of $v$ sending the ID of its cluster center to $w$ in a first round, and receiving a reply with $f_v(w)$ in the second round (where $w$ checks if its cluster center's ID is the same as that of $v$'s ). Now, note that if $v$ has no incident outgoing edges, then $\varepsilon_f = |f_v^{-1}(1)| / deg(v) = 0$, otherwise $\varepsilon_f = |f_v^{-1}(1)| / deg(v) \geq 1 / deg(v)$. As $\eps_f \geq \eps$, by Theorem \ref{dgs}, node $v$ outputs some node $w \in V$ with $f_v(w) = 1$ with probability $1-\alpha=1-1/n^3$. In other words, $v$ outputs some incident edge leading outside $C$ with probability $1-1/n^3$. Thus, by a union bound over all $n$ nodes, the following statement holds with probability at least $1- 1/n^2$: any node $v \in V$ finds an outgoing incident edge, if there exists any. After which, in the second part, the convergecast transmits at least one outgoing incident edge (if there exists any) to the cluster center.

    Next, we bound the round and message complexities. By Theorem \ref{dgs}, each Grover Search has round and message complexity $O(\sqrt{\deg(v)} \cdot \log n)$. Hence, the first part has round complexity $O(\sqrt{n} \cdot \log n)$ and message complexity $O(\sum_{v \in V} \sqrt{\deg(v)} \cdot \log n) = O(\sqrt{m n} \log n)$ (by the Cauchy-Schwarz inequality). As for the second part, the convergecast takes round and message complexity $O(n)$ (as the cluster tree may have up to $n$ depth). Finally, it suffices to add up the round and message complexities.
\end{proof}

The following statement for steps (2) and (3) follows directly from the classical case, and we omit its proof.

\begin{lemma}
    Step (2) and (3) take $O(n \log^* n)$ rounds and sends $O(n \log^* n)$ messages. Moreover, these two steps guarantee that the phase ends with at most half of the clusters it starts with.
\end{lemma}

It follows that after $O(\log n)$ phases, there remains a single cluster with high probability. The correctness of \quantumGeneralLE{} follows by a simple union bound (on the phases), and the round and message complexity upper bounds are straightforward.

\begin{theorem}
    \quantumGeneralLE{} solves explicit leader election with probability at least $1-1/n$.  Moreover, it takes $\tilde{O}(n)$ rounds and sends $\tilde{O}(\sqrt{m n})$ messages.
\end{theorem}

\section{Quantum  Agreement in Complete Networks}
\label{sec:agreement}
We describe a quantum protocol for implicit agreement, assuming shared randomness (see \Cref{sec:pbs}) in complete networks of $n$ nodes.The presented protocol is a quantum boosting of the classical protocol of~\cite{AMP18}. The key differences lie in how we (1) estimate how many nodes have a certain input (or vote), and (2) detect when agreement is reached. For both, we redesign the classical protocol to leverage quantum subroutines and obtain quadratic factor improvements in the message complexity. 

Indeed, estimating the number of nodes with input (or vote) 1 within some $n$-sized universe up to some $\eps n$ additive error, for any $\eps < 1$ (and say, with constant success probability) requires sending $\Omega(1/\eps^2)$ messages in the classical setting, whereas approximate quantum counting allows us to send only $O(1/\eps)$ messages. As for detecting when agreement is reached, Grover search allows us (as with the handshake problem) to get a quadratic factor reduction in the message complexity when compared to the classical setting. Finally, roughly speaking, we balance the improved complexities of these two parts to obtain a quadratic factor improvement in the message complexity over the classical protocol  of~\cite{AMP18}, but our use of these quantum subroutines come at the cost of increased runtime. 

\subsection{Algorithm Description}

The implicit agreement protocol \quantumAgreement{} (\Cref{alg:shared-rand-i-algo})
 is separated into two phases: an estimation phase, followed by an agreement phase.
The protocol is parametrized by some variable $\eps \in [\Theta(1/n),1/20]$ and constant $ \gamma \in [0, 1/3]$, which allows for a trade off between rounds and messages. Of particular interest is when $\eps = 1/n^{1/5}$ and $\gamma = 2/15$, which optimizes the message complexity at $\tilde{\Theta}(n^{1/5})$.

\paragraph{Estimation Phase.} In this phase, which is essentially quantum, each node becomes a \emph{candidate} with probability $p = (\pCst \ln n)/ n$.
After which, each candidate node $v$ computes an estimation $q(v)$, up to $\eps$ additive error, of the fraction (denoted by $q$) of $1$s in the network. To do so, $v$ runs the distributed quantum approximate counting primitive  $\mathsf{ApproxCount}(\eps,\alpha_1)$ (\Cref{cor:approxCount}), with $\alpha_1 = 1/(2 n^2)$ error, to estimate the number of nodes with initial input 1 (for agreement), up to $\eps n$ additive error.
In this part, the distributed algorithm $\checking_g$ computes the function $g : V \rightarrow \{0,1\}$ 
defined such that $g(w) = 1$ if and only if $w$'s initial input (for agreement) is $1$, for $w \in V$. (Clearly, $\checking_g$ runs for 2 rounds and sends 2 messages.) After which, $v$ divides the output of the primitive by $n$ (which is their degree plus one) to obtain $q(v)$.

\paragraph{Agreement Phase.} In this second phase, nodes run a while loop for $\ell = O(\log n)$ iterations, at the end of which all nodes (including non candidate ones) terminate. Each iteration of that loop --- called \emph{agreement iteration} --- starts with a classical part, followed by a quantum part. The output of the quantum part (consistent among all candidates nodes with high probability) determines whether the candidates nodes exit the while loop and terminate early, or not.

In the first, classical part, the candidate nodes choose a shared random value $r$ uniformly at random in $[0,1]$ (via the shared randomness available to them). Then, each candidate node $v$ is undecided if $|q(v)-r| \leq \eps$, and otherwise becomes decided, choosing value 0 (resp., 1) if $q(v) < r - \eps$ (resp., $q(v) > r + \eps$). After which, each decided node sends (classically) a message (containing its value) to $O(n^{1/3-\gamma})$ neighbors, chosen arbitrarily, which sets up the quantum part.

In the second, quantum part, each undecided (candidate) node $u$ checks the existence (or not) of a decided node via (distributed) Grover Search. To do so, $u$ runs $\mathsf{GroverSearch}(\eps_2,\alpha_2)$ with $\eps_2 = n^{-2/3-\gamma}$, $\alpha_2 = 1/(4 n^3)$ and where (the distributed algorithm) $\checking_h$ computes function $h : V \rightarrow \{0,1\}$, where for any node $w \in V$, $h(w) = 1$ if and only if $w$ received a message from some decided node in the classical part (of this iteration). (Clearly, $\checking_h$ runs for 2 rounds and sends 2 messages.) After which, decided nodes terminate, while undecided nodes terminate or not depending on the output of the Grover Search. More concretely, if undecided node $u$ finds a node $w \in V$ such that $h(w) = 1$, then node $u$ terminates the agreement protocol. Else, $u$ starts the next iteration of the while loop. 

\begin{algorithm}[ht]\small
\caption{(\quantumAgreement{}) Quantum agreement protocol for complete networks}
\label{alg:shared-rand-i-algo}

\begin{algorithmic}[1]

\Require A complete $n$-node anonymous network where each node receives a value in $\{0, 1\}$ given by an adversary. The nodes have access to an unbiased global coin that is oblivious to the adversary.
\Ensure Implicit agreement.
\smallskip
\item[\textbf{Estimation Phase:}]
\State (Classical) Each node elects itself as a candidate node with probability $\frac{12\ln n}{n}$.
 
\State \textbf{Counting Step (Quantum):} Each candidate node $v$ runs $\mathsf{ApproxCount}( \eps, 1/(2 n^2))$  to estimate the number of nodes with initial input 1 (for agreement), up to $\eps n$ additive error.

\State (Classical) Each candidate node $v$ divides the output of the counting primitive by $n$ to obtain $q(v)$ which is an estimate of the fraction of $1$s in the network, up to $\eps$ additive error.

\smallskip

\item[\textbf{Agreement Phase:}] 

\For{$i =1, 2, \dots$}

\State (Classical) Each candidate node uses the  global coin to generate a (shared) random number $r$ in $[0,1]$. \label{stp:verification} 
\State  (Classical) Each candidate node $v$ is {\em undecided} if $|q(v) - r| \leq  \eps$. 
Otherwise  it becomes decided --- choosing value 0 (resp., 1) 
if $q(v) < r - \eps$ (resp., $q(v) > r + \eps$).

\State (Classical) Each decided node sends  a message (containing its value) to $O(n^{1/3-\gamma})$ neighbors, chosen arbitrarily.
\smallskip

\State \textbf{Verification step (Quantum):}
 Each undecided candidate node $u$ runs $\mathsf{GroverSearch}(n^{-2/3-\gamma},1/(4 n^3))$ to check for the existence (or not) of a decided node.
\If{undecided node $u$ finds a decided node $w \in V$} 
\State Stop and exit from the `for-loop'   
\EndIf \label{stp:verification-end}
\EndFor
\State All the candidate nodes (both decided and undecided) know the $deciding\_value$. 
\end{algorithmic}
\end{algorithm}

\subsection{Analysis}

Let $C$ the set of candidate nodes and $q$ the fraction of $1$s in the network. Note that by \Cref{obs:sampleAndUniqueRanks} in Subsection \ref{subsec:randomizedPrelims}, we know $1 \leq |C| \leq \pMax \ln n$ with probability at least $1-1/n^2$. 

We start with a lemma capturing the guarantees of the estimation phase (roughly speaking, its correctness, time and message complexities). Moreover, let $\est$ be the event that the candidates' estimates of the fraction of $1$s in the network are within an $\eps$ additive error. We also show that $\est$ happens with high probability, and independently of (the random variable) $|C|$.

\begin{lemma}
\label{lem:candidateEstimates}
    The estimation phase takes $\tilde{O}(1/\eps)$ rounds. Moreover, $\tilde{O}(1/\eps)$ messages are sent both in expectation and with probability at least $1-1/n$. Finally, $\est$ happens with probability at least $1-1/(2n)$, and independently of $|C|$.
\end{lemma}

\begin{proof}
    In the estimation phase, each candidate node $v$ runs $\mathsf{ApproxCount}(\eps,\alpha_1)$, with $\alpha_1 = 1/(2 n^2)$ and function $g$. 
    From the definition, $|g^{-1}(1)|$ is the number of $1$s in the network and thus $|g^{-1}(1)| = q n$. Then, by \Cref{cor:approxCount}, the output $n(v)$ satisfies $|n(v)- q n| \leq \eps n$ with probability at least $1-\alpha_1$. Since the output is then divided by $n$, $v$ computes $q(v)$ such that $|q(v)-q| \leq \eps$ with probability at least $1-\alpha_1$. By a simple union bound over the (at most $n$) candidate nodes, it follows that all candidates nodes' estimates lie in $[q-\eps,q+\eps]$ with probability at least $1-1/(2 n)$. Moreover, this event happens independently of any probability statement on $|C|$, due to the crude union bound.

    Now, we prove the time and message complexities. By \Cref{cor:approxCount}, $\mathsf{ApproxCount}(\eps,\alpha_1)$ takes $O(\log(1/\alpha_1)/\eps) = \tilde{O}(1/\eps)$ rounds and messages. Since there are $O(\log n)$ candidates with probability at least $1-1/n^2$, and at most $n$ otherwise, the (expected and with high probability) message complexity upper bounds follows.
\end{proof}

Next, we bound the round complexity and expected message complexity of the agreement phase (see Lemma \ref{lem:candidateAgreement}). For that purpose, we prove the following auxiliary lemma (see Lemma \ref{lem:NoDecidedNode}) showing that any agreement iteration contains undecided candidate nodes with probability at most $O(\eps)$, when $\est$ holds true. The obtained probability crucially balances out the fact that any undecided candidate runs a Grover Search that sends $\tilde{O}(n^{1/3+\gamma/2})$ messages, which leads to our claimed expected message complexity bounds.

\begin{lemma}
\label{lem:NoDecidedNode}
    Given $\est$ holds true, then for any agreement iteration, all candidate nodes become decided (equivalently, no candidate node is undecided) with probability at least $1 - 4 \eps$.
\end{lemma}

\begin{proof}
    In any agreement iteration, the candidate nodes choose a shared random value $r$ uniformly at random in $[0,1]$. Now, all candidate nodes become decided in this iteration (or in other words, none is undecided) if for any candidate node $v$, it holds that $|q(v)-r| > \eps$.
    Since $\est$ holds true, all candidate estimates are contained in $[q-\eps,q+\eps]$. As a result, all candidate nodes become decided (in this iteration) if $|q-r| > 2\eps$, or in other words, if $r$ is not chosen in a strip of length $4 \eps$ centered at $q$. This happens with probability at least $1- 4\eps$.
\end{proof}

\begin{lemma}
\label{lem:candidateAgreement}
    The agreement phase takes $\tilde{O}(n^{1/3+\gamma/2})$ rounds. Moreover, it uses in expectation $\tilde{O}(n^{1/3-\gamma}+\eps \cdot n^{1/3+\gamma/2})$ messages.
\end{lemma}

\begin{proof}
    By the algorithm description, the agreement phase takes $\ell = O(\log n)$ iterations. Hence, in what follows it suffices to bound the round and expected message complexity of each iteration. We first bound these for the classical part. In it, each decided candidate node sends $O(n^{1/3-\gamma})$ messages in a single round. By \Cref{obs:sampleAndUniqueRanks}, there are $O(\log n)$ candidates with probability at least $1-1/n^2$. Thus, for each iteration, the classical part takes $1$ round (deterministically) and $\tilde{O}(n^{1/3-\gamma})$ messages (in expectation).

    It remains to bound the complexity for the quantum part. 
    By Theorem \ref{dgs}, each undecided node's Grover Search takes $\tilde{O}(1/\sqrt{\eps_2}) = \tilde{O}(n^{1/3+\gamma/2})$ rounds and messages. This bounds the round complexity of the quantum part, but we now need to bound its (expected) message complexity, which we do by bounding the expected number of undecided candidate nodes. First, by \Cref{obs:sampleAndUniqueRanks}, the number of candidates is $O(\log n)$ with probability at least $1-1/n^2$, and at most $n$ otherwise (the latter event is denoted by $A$). Moreover, it holds independently by Lemma \ref{lem:candidateEstimates} that $\est$ holds true with probability at least $1-1/n$, in which case Lemma \ref{lem:NoDecidedNode} implies that there are no undecided candidate nodes (independently of how many nodes are candidates) with probability at least $1 - 4 \eps$. We denote this last event, leading to no undecided candidates nodes, by $B$. We now consider three mutually exclusive events:
    \begin{itemize}
        \item If $A$ holds --- which happens with probability $p_1 = 1/n^2$ --- there are at most $n$ undecided candidates,
        \item If $\bar{A} \cap B$ holds, there are no undecided candidates,
        \item Otherwise --- which happens with probability at most $p_2 = (1-1/n^2) (1/n + (1-1/n) (4\eps))$ --- there are at most $O(\log n)$ undecided candidates.
    \end{itemize}  
    Now, it is clear that the expected number of undecided candidates per iteration is at most 
    $$p_1 n + p_2 \cdot O(\log n) = \tilde{O}(\eps)$$  where we use the fact that $\eps = \Omega(1/n)$.
    Therefore, the expected message complexity of any iteration's quantum part is $\tilde{O}(\eps \cdot n^{1/3+\gamma/2})$. 
\end{proof}

Now, we can show that if $\est$ holds true (which happens with high probability), the agreement phase produces a valid (implicit) agreement output with high probability (see Lemma \ref{lem:agreementPhaseCorrectness}). To do so, we first prove two auxiliary lemmas. The first --- see Lemma \ref{lem:agreementCorrectness} --- implies that any agreement iteration in which there is some decided candidate node leads to the termination of all candidate nodes with high probability, when $\est$ holds true. The second --- see Lemma \ref{lem:noDisagreement} --- proves that any two candidates nodes that become decided in the same agreement iteration must agree on their decided value. 

\begin{lemma}
\label{lem:agreementCorrectness}
    Let $\est$ hold true. For any agreement iteration, 
    if at least one candidate is decided, then with probability $1-1/(4 n^2)$, all undecided candidate nodes detect a decided node and all candidate nodes (both undecided and decided) terminate.
\end{lemma}

\begin{proof}
    Any agreement iteration starts with all decided candidate nodes sending messages to $n^{1/3-\gamma}$ neighbors during the classical part. 
    After which, undecided nodes use Grover Search during the quantum part of that same iteration, with parameters $\eps_2 = n^{-2/3-\gamma}$, $\alpha_2 = 1/(4n^3)$ and function $h$. By definition of $h$, if there are no decided nodes after the classical part of this iteration, then $\eps_f = |h^{-1}(1)|/n = 0$, and otherwise $\eps_f = |h^{-1}(1)|/n \geq n^{1/3-\gamma}/n = n^{-2/3-\gamma}$. As $\eps_f \geq \eps_2$, by Theorem \ref{dgs}, any undecided node $u$ outputs some node $w \in V$ such that $h(w) = 1$ with probability at least $1-\alpha_2$. 
    By a simple union bound over the (at most $n$) candidate nodes, all undecided nodes detect the existence of some decided node with probability at least $1-1/(4n^2)$, and thus terminate the agreement protocol at the end of this agreement iteration.
\end{proof}

\begin{lemma}
\label{lem:noDisagreement}
    If $\est$ holds true, then any two candidate nodes that become decided in the same agreement iteration agree on the same value.
\end{lemma}

\begin{proof}
    By contradiction, let $u$ and $v$ be two candidate nodes that decide respectively on 0 and 1, in some agreement iteration.
    Since $\est$ is true, both estimates $q(u)$ and $q(v)$ are contained in $[q-\eps,q+\eps]$ by the end of the estimation phase. However, by the algorithm description, these two estimates must satisfy $q(u) < r - \eps$ and $q(v) > r + \eps$, where $r$ is the shared random value chosen in this agreement iteration. Since this implies $|q(v) - q(u)| > 2 \eps$, we obtain a contradiction.
\end{proof}

\begin{lemma}[Valid Agreement Output]
\label{lem:agreementPhaseCorrectness}
    If $\est$ holds true, then with probability at least $1-1/(2n)$, at least one node becomes decided within the $\ell = O(\log n)$ agreement iterations, and all decided nodes agree on the same value.
\end{lemma}

\begin{proof}
    Since $\est$ holds true, then it holds by Lemma \ref{lem:NoDecidedNode} that for any agreement iteration, all nodes (and thus at least one) become decided with probability at least $1- 4\eps$. Since the probability of all nodes becoming decided is independent over different iterations, we get that no candidate node becomes decided within the $\ell$ phases with probability at most $(4 \eps)^{\ell}$. Since $\eps \leq 1/20$, then $(4 \eps)^{\ell} \leq (1/5)^{\ell} \leq 1/(4n)$ for a well-chosen $\ell = O(\log n)$. 
    
    Next, we show that all decided nodes agree on the same value. By Lemma \ref{lem:agreementCorrectness}, with probability at least $1-1/(4n^2)$, for any agreement iteration in which there exists at least one decided node, all candidate nodes terminate by the end of that iteration. A union bound shows that with probability at least $1-1/(4n)$, for any agreement iteration (within the first $n$) in which at least one node becomes decided then all candidate nodes terminate by the end of that iteration. This implies that all nodes that become decided do so in the same agreement iteration, with probability at least $1-1/(4n)$. And thus, by Lemma \ref{lem:noDisagreement}, all nodes that become decided agree on the same value, with probability at least $1-1/(4n)$.
    
    Thus, by the end of the $\ell$ iterations, all (candidate and non-candidate) nodes terminate, and all decided nodes agree on the same value, with probability at least $1-1/(2n)$.
\end{proof}

Finally, we can prove our main result regarding our quantum implicit agreement protocol. Note that the theorem below is parametrized by $\eps \in [\Theta(1/n),1/20]$ and $ \gamma \in [0, 1/3]$.

\begin{theorem}
    \quantumAgreement{} solves implicit agreement with probability at least $1-1/n$. Moreover, the protocol takes $\tilde{O}(1/\eps + n^{1/3+\gamma/2})$ rounds and sends in expectation $\tilde{O}(1/\eps+n^{1/3-\gamma}+\eps \cdot n^{1/3+\gamma/2})$ messages.
\end{theorem}

\begin{proof}
    We start with the correctness. First, note that by Lemma \ref{lem:candidateEstimates}, the estimation phase is successful with probability at least $1-1/(2n)$. In which case, by Lemma \ref{lem:agreementPhaseCorrectness}, the agreement protocol terminates with a correct (implicit agreement) output with probability at least $1-1/(2n)$. Hence, the protocol solves implicit agreement with probability at least $1-1/n$.
    
    Next, we bound the time and message complexities. First, by Lemma \ref{lem:candidateEstimates}, the estimation phase takes $\tilde{O}(1/\eps)$ rounds, and $\tilde{O}(1/\eps)$ messages in expectation.
    Second, by Lemma \ref{lem:candidateAgreement}, the agreement phase takes $\tilde{O}(n^{1/3+\gamma/2})$ rounds, and sends $\tilde{O}(n^{1/3-\gamma}+\eps \cdot n^{1/3+\gamma/2})$ messages in expectation. It suffices to add the complexities of the two phases together.
\end{proof}

By setting $\eps = 1/n^{1/5}$ and $\gamma = 2/15$, we can optimize the message complexity in the above statement. This yields a quantum implicit agreement protocol with expected message complexity $\tilde{O}(n^{1/5})$. In contrast, the best-known agreement protocol in the classical setting has a quadratically worse expected message complexity of $O(n^{2/5})$ \cite{AMP18}.

\begin{corollary}
    \quantumAgreement{} solves implicit agreement with probability at least $1-1/n$. Moreover, the protocol takes $\tilde{O}(n^{2/5})$ rounds and sends in expectation $\tilde{O}(n^{1/5})$ messages.
\end{corollary}

\section{Conclusion and Open Questions}

In this work, we presented new quantum distributed algorithms for leader election and
agreement that have message complexities that significantly improve over their respective
classical counterparts. To design and analyze these algorithms, we showed how  quantum algorithmic techniques such as Grover search, quantum counting, and quantum walks can be used in the distributed setting to achieve improvement in message complexity.  

Our work raises several key open questions. While our results use quantum subroutines initially developed to study sequential algorithms for time and query complexities, we were not able to export existing techniques for proving quantum query lower bounds to establish lower bounds on the message complexity of our protocols. Indeed, the reduction is more in the opposite direction, letting us export the existing framework to distributed computing. 

Nonetheless, we \emph{conjecture} all our protocols are \emph{tight}, except for quantum leader election in networks with mixing time $\tau$, for which we suspect that achieving a message complexity linear in $\tau$ may be possible.

We think that establishing non-trivial lower bounds on the quantum message complexity for leader election, and for specific network configurations, could require the development of new techniques that could be of broader interest, particularly for a better understanding of the limitations of quantum computing and quantum information in general.

Last, our improvements to message complexity come at the cost of a significantly increased number of rounds. Still, one can change our parameters to reduce the number of rounds and still get a message complexity that goes below the classical setting's best achievable message complexity.
It would be interesting to see how this trade-off could be improved.

\bibliographystyle{alpha}
\bibliography{biblio}

\appendix

\section{Supplementary Materials for Non-Oblivious Quantum Distributed Computing}

\subsection{Formal Model for Quantum Routing}\label{app:formal}
We now give a formal model for quantum routing while borrowing some notations of~\cite{ck19}.
Similar to the classical setting (see \Cref{sec:model}), each node $u$ has $\deg(v)$ ports, one for each node $v$ connected to $u$. To each such a port, $u$ holds two registers ${u\to v}$ and ${u\gets v}$. With some abuse of notations we denote a port $p$ by the pair $p=(u,v)$, meaning that this is a port based in $u$ for communicating with $v$. 

Below are the different possible states of the registers associated to the communication from $u$ to $v$, that is $u\to v$ from the viewpoint of $u$, and $v\gets u$ from the viewpoint of $v$.
Note that a node can be in a superposition of all those states. We denote by $\ket{\vac}$ a default state at the beginning of the round, representing the vacuum state, which is orthogonal to any possible message $\ket{m}$, so that one can distinguish between a request to send a message or a message has been received, and the absence of any message.
\begin{itemize}
\item End of round $i$:
\begin{itemize}
\item Register ${u\to v}$ of $u$: $\quad\begin{cases}\ket{\vac}_{u\to v},& \text{when no message to send to $v$}\\ \ket{m}_{u\to v},& \text{when message $\ket{m}$ needs to be sent to $v$}\end{cases}$
\item Register ${v\gets u}$ of $v$: $\quad\ket{\vac}_{v\gets u}$, \quad for preparing the potential reception \end{itemize}
\item Beginning of round $i+1$:
\begin{itemize}
\item Register ${u\to v}$ of $u$: $\quad\ket{\vac}_{u\to v}$, \quad since the message has been delivered
\item Register ${v\gets u}$ of $v$: $\quad\begin{cases}\ket{\vac}_{v\gets u},& \text{ when no message has been sent by $u$}\\ \ket{m}_{v\gets u},& \text{when message $\ket{m}$ has been sent by $u$}\end{cases}$
\end{itemize}
\end{itemize}

The $\send_{u\to v}$ operation models message emission from $u$ to $v$ by swapping register  ${u\to v}$ in $u$ with register  ${v\gets u}$ in $v$: 
$$\send_{u\to v} (\ket{m}_{u\to v}\ket{\vac}_{v\gets u})= \ket{\vac}_{u\to v}\ket{m}_{v\gets u}.$$

Then, the emission operator $\send_u$, which handles all messages from $u$, is simply defined as the tensor product of those operations for each port, and similarly for $\send$ for the emission of all messages sent across the network:
$$\send_u = \bigotimes_{p=(u,v)}  \send_{u\to v}  \quad \text{ and } \quad \send =\bigotimes_{u}\send_u.$$

\subsection{Example}  
\label{app:example}
Let us see how to send a message $\ket{m}$ to port (and therefore a node) selected according to the
superposition $\ket{\psi}=\sum_{p=(u,v)} \alpha_p \ket{p}$. 
Since some $\alpha_p$ could be $0$, this models a node selecting a subset of ports.
As we will see in \Cref{sec:quantization}, this encompasses the case of the randomized selection of one port. In particular, when only one port $p$ satisfies $\alpha_p\neq 0$, the selection is deterministic.
We show how to prepare the sending registers (1) and we explicit the action of the operator $\send_u$ (2).

Below, we put a subscript $u$ to indicate that the corresponding register is local to $u$ and different than its emission/reception registers.
We also only explicit the registers impacted by $\send_u$ . 
\begin{enumerate}
    \item Message preparation by $u$ (using unitaries called control-swap), where all registers are located in $u$:
\begin{multline*}
\Big(\ket{m}\otimes \Big(\sum_{p=(u,v)} \alpha_v \ket{v}\Big)\Big)_u\bigotimes_{p=(u,w)} \ket{\vac}_{u\to w} \\
\mapsto  \sum_{p=(u,v)} \alpha_v \Big( \ket{\vac} \otimes \ket{v} \Big)_u \otimes\ket{m}_{u\to v}\bigotimes_{p=(u,w\neq v)} \ket{\vac}_{u\to w}.
\end{multline*}
    \item Action of $\send_u$, where registers $v\gets u$ and $z\gets u$ belongs respectively to $v$ and $z$:
\begin{multline*}
\sum_{p=(u,v)} \alpha_v \Big( \ket{\vac} \otimes \ket{v} \Big)_u \otimes\ket{m}_{u\to v}\bigotimes_{p=(u,w\neq v)} \ket{\vac}_{u\to w} \bigotimes_{p=(z,u)} \ket{\vac}_{z\gets  u}\\
    \mapsto
    \sum_{p=(u,v)} \alpha_v \Big( \ket{\vac} \otimes \ket{v} \Big)_u 
    \bigotimes_{p=(u,w)} \ket{\vac}_{u\to w} \otimes \ket{m}_{v\gets u} \bigotimes_{p=(z\neq v,u)} \ket{\vac}_{z\gets  u}.
\end{multline*}
\end{enumerate}
In this example, one can see that, in each term of the superposition, the message $m$ has been sent to a single node $v$.

\section{Supplementary Materials for Distributed Quantum Subroutines}

\subsection{Quantization of (Distributed) Algorithms} \label{app:quantization}
We review and detail some of the arguments already presented in~\cite{LM18}. Indeed, we are in slightly more refined model because we now also measure the message complexity. It consists in transforming the classical algorithm, deterministic or randomized, instruction per instruction. Propagating those transformations leads to a fully quantum algorithm that is described by a unitary map.

First, for the communication part between each round, one way to model it is to assume that it consists of an exchange of registers, one of them being in a default state that we call a vacuum.

Second, for the local computation in each node, let us first review a simple transformation if we do not care about the computational power of the nodes. Observe that any local deterministic operation of a node is just a function from states to states. A randomized one is more generally a stochastic transformation that can be represented by a stochastic matrix $P$. The transformation $P$ can always be \emph{purified} into a unitary map $U$ as follows. The transformation introduces a second state register initialized to some default configuration, say the $0$-string. Then $U$ is defined as $\ket{x,y}$ to $\sum_{z} \sqrt{P_{xy}}\ket{x,y\oplus z}$, where $\oplus$ denotes the bit-wise XOR. Using $U$ with $y$ set to the $0$-string leads to the distribution we wanted if one observes the second register. We just have to postpone this observation to the end of the algorithm to keep the process unitary. Then, the computation continues with the second register, whereas the first one is kept to make the transformation unitary.

Last, let us add that such a transformation can in fact preserve the computational power of the nodes, while it is done at each local algorithmic instruction. This can even be optimized further
using now standard techniques coming back to the early age of reversible computing and initiated by Bennett~\cite{Bennett+SICOMP89}. This may lead to additional memory space, for instance, that is proportional to the running time. But other strategies exist depending on whether one would preserve either time or space complexity. In the context of quantum computing, there also exists several techniques for differing intermediate measurements; see, for instance,~\cite {GirishR22}.

\subsection{Example on a Star Graph}\label{ref:starexample}
Consider a star graph of $n+1$ nodes. Call the center node $u$, and the other nodes $X$. Each node gets in $X$ an input bit $b_v$ in $\{0,1\}$.

In the Searching problem, the center node $u$ wants to find a node $v$ with bit $1$, that is such that $b_v=1$. 
In the Counting problem, $u$ wants  to approximate the number of nodes $v$ with $b_v=1$.

\para{Searching.}
Classically, the message complexity is $\Theta(n)$ in $1$ round. This cannot be improved with more rounds. Instead, we will see that with more rounds, one can decrease the \emph{quantum} message complexity.

Define $f:X\to\{0,1\}$, such that $f(v)=b_v$.
Clearly, $\checking$ costs $2$ bits of communication and $2$ rounds: $1$ bit from $u$ to $v$ in order to ``ask'' as a return the bit $b_v$, which is the $2$nd bit of communication. Thus by \Cref{dgs}, the problem can be solved using $O(\sqrt{n})$ rounds and $O(\sqrt{n})$ bits of communication.

In order to reduce the number of rounds while increasing the message complexity, one could partition arbitrarily nodes of $X$ into buckets $X_1,X_2,\ldots,X_{\lceil n/k\rceil }$ of size at most $k$. Then, now $f(i)$ is the OR of the bits in $X_i$, and $\checking$ costs $2$ rounds and $2k$ bits of communication, as there are $\lceil n/k\rceil$ buckets.
Thus by \Cref{dgs}, finding a bucket with a $1$ can be solved using $O(\sqrt{n/k})$ rounds and $O(\sqrt{n/k}\times k)=O(\sqrt{nk})$ bits of communication. After which, finding the vertex with a $1$ within a bucket requires additional $k$ bits of communication and $1$ round.

\para{Counting.}
Classically, to approximate the number of nodes $v$ such that  $b_v=1$ with additive error $\eps n$,
 the message bits complexity is $O(1/\eps^2)$ in $1$ round, and this cannot be improved with more rounds.
Quantumly, with the same definition of $f$ as before, but using \Cref{cor:approxCount}, the center node $u$ can approximately count within $O(1/\eps)$ rounds and $O(1/\eps)$ message complexity.

\section{Randomized Tools}
\label{subsec:randomizedPrelims}

First, we provide Chernoff bounds~\cite{MitzenmacherUpfalBook} for the sum of independent Bernoulli random variables.

\begin{lemma}
\label{lem:chernoffBoundBernoulli}
Let $X$ be the sum of $k$ independent variables with $\{0,1\}$ value, with expectation $\E[X] = \mu$. Then, for any $\delta \geq 0$, $\Pr[X \geq (1+\delta) \mu] \leq \exp(-\frac{\delta^2 \mu}{(2+\delta)})$.
\end{lemma}

Given these Chernoff bounds, we can give the following observation on sampling in distributed networks, combined with an observation on choosing unique ranks. More concretely, suppose that all nodes in the (communication) network $G = (V,E)$ sample themselves with probability $p = (\pCst \ln n)/ n$, and also that they each choose a rank independently and uniformly at random in $\{1,\ldots,n^4\}$. (This assumes nodes know $n$. For a multiplicative constant approximation of $n$, the probability can be set accordingly.) 

\begin{fact}
\label{obs:sampleAndUniqueRanks}
    With probability at least $1-1/n^2$, it holds for any $n \geq 2$ both that
    \begin{itemize}
        \item at least one, and at most $\pMax \log n$, nodes are sampled,
        \item no two nodes choose the same rank.
    \end{itemize}
\end{fact}

\begin{proof}
    Recall that each node is sampled with probability $p = (\pCst \ln n) / n$, thus there are in expectation $\mu = p \, n = \pCst \ln n$ sampled nodes. Let us denote the number of sampled nodes by the random variable $X$. Since node are sampled independently, the probability that no node is sampled is at most $\Pr[X = 0] = (1-p)^n \leq \exp(- p \, n)$, where the inequality follows from the fact that for any $x \in \mathbb{R}$, $1+x \leq e^x$. Since $- \pCst \ln n \leq - \ln 4 - 2 \ln n$ for any $n \geq 2$, we get that no node is sampled with probability at most $(1/4) n^{-2}$. This lower bounds the number of sampled nodes. For the upper bound, a simple application of Chernoff bounds with $\delta = 1$ (see Lemma \ref{lem:chernoffBoundBernoulli}) shows that $\Pr[X \geq 2 \mu] \leq \exp(- \mu / 3)$. For any $n \geq 2$, it holds that $- \mu/3 \leq - \ln 4 - 2 \ln n$, thus  $\Pr[X \geq \pMax \ln n] \leq (1/4) n^{-2}$. 

    Finally, recall that all nodes choose ranks independently and uniformly at random in $\{1,\ldots,n^4\}$. Thus, the probability that any two nodes have the same rank is at most $n^{-4}$. Since there are at most $n^2/2$ unordered pairs of nodes, by a union bound we get that no two nodes have the same rank with probability at most $(1/2) n^{-2}$. Combining all three error probabilities adds up to an $n^{-2}$ error probability, and thus we get the lemma statement. 
\end{proof}

\end{document}